\newif\ifextended
\newif\ifcomments

\extendedtrue

\commentsfalse

\ifextended
\documentclass[acmsmall,nonacm,screen]{acmart}
\else
\documentclass[acmsmall]{acmart} 
\fi

\setcopyright{rightsretained}
\acmJournal{PACMPL}
\acmYear{2025}
\acmVolume{9}
\acmNumber{POPL}
\acmArticle{20}
\acmMonth{1}
\acmDOI{10.1145/3704856}

\setcitestyle{nosort}

\usepackage{booktabs}   
\usepackage{subcaption} 

\usepackage[utf8]{inputenc}
\usepackage[english]{babel}
\usepackage{xparse}
\usepackage{xspace}
\usepackage{thmtools}

\ifcomments
\usepackage[draft]{fixme} 
\else
\usepackage[final]{fixme} 
\fi

\usepackage{xcolor}
\usepackage{mathpartir}
\usepackage[]{microtype}
\usepackage{multirow}
\usepackage{makecell}
\usepackage{marvosym}
\usepackage{wasysym}
\usepackage{pifont}
\usepackage{mathtools}
\usepackage{stmaryrd}
\usepackage{scalerel}
\usepackage{tensor}
\usepackage{xifthen}
\usepackage{iris}
\usepackage{pftools}
\usepackage{soul}
\usepackage{tikz}
\usepackage{threeparttable}
\usepackage{fontawesome}
\usepackage{enumitem}
\usetikzlibrary{calc, shapes, arrows, automata, patterns}

\makeatletter
\addto\extrasenglish{%
  \renewcommand*\chapterautorefname{\S\@gobble}
  \renewcommand*\sectionautorefname{\S\@gobble}
  \renewcommand*\subsectionautorefname{\S\@gobble}
  }
\makeatother
\newcommand*{\myeqref}[2][]{%
  \hyperref[{#2}]{#1(\ref*{#2})}%
}

\newcommand*{\lineref}[1]{\hyperref[#1]{line~\ref*{#1}}}
\newcommand*{\linerangeref}[2]{\hyperref[#1]{lines~\ref*{#1}}\hyperref[#2]{-\ref*{#2}}}
\newcommand{\mypageref}[1]{\hyperref[#1]{page~\pageref*{#1}}}

\newcommand{\appendixref}[2]{\ifextended{\autoref{#1}}\else{the extended version~\cite[Section #2]{extendedversion}}\fi}

\declaretheorem[name=Definition,style=definition]{definition}
\declaretheorem[name=Theorem,sibling=definition]{theorem}

\DeclareUnicodeCharacter{2264}{$\le$}
\DeclareUnicodeCharacter{2208}{$\in$}
\DeclareUnicodeCharacter{2203}{$\exists$}
\DeclareUnicodeCharacter{25C1}{$\triangleleft$}
\DeclareUnicodeCharacter{2097}{${}_l$}
\DeclareUnicodeCharacter{2260}{$\neq$}
\DeclareUnicodeCharacter{2205}{$\emptyset$}
\DeclareUnicodeCharacter{222A}{$\cup$}
\DeclareUnicodeCharacter{228E}{$\uplus$}
\DeclareUnicodeCharacter{2200}{$\forall$}
\DeclareUnicodeCharacter{25A1}{$\always$}
\DeclareUnicodeCharacter{03B3}{$\gamma$}

\definecolor{airforceblue}{rgb}{0.36, 0.54, 0.66}
\definecolor{brickred}{rgb}{0.9, 0.55, 0}
\definecolor{ao}{rgb}{0.0, 0.0, 1.0}
\definecolor{cobalt}{rgb}{0.0, 0.28, 0.67}
\definecolor{darkergreen}{rgb}{0,0.7,0.3}
\definecolor{magenta}{rgb}{1.0,0.0,1.0}

\FXRegisterAuthor{td}{atd}{TD}
\FXRegisterAuthor{ms}{ams}{MS}
\FXRegisterAuthor{gp}{agp}{GP}
\FXRegisterAuthor{as}{aas}{AS}
\FXRegisterAuthor{pm}{apm}{PM}

\makeatletter
\providecommand*{\Dashv}{%
  \mathrel{%
    \mathpalette\@Dashv\vDash
  }%
}
\newcommand*{\@Dashv}[2]{%
  \reflectbox{$\m@th#1#2$}%
}
\makeatother
\makeatletter 
\def\arcr{\@arraycr}
\makeatother

\newcommand\ie{\emph{i.e.}, }
\newcommand\eg{\emph{e.g.}, }

\makeatletter
\def\@parfont{\bfseries\itshape}
\makeatother

\newcommand{\theIVL}{CoreIVL}
\newcommand{\ViperInst}{ViperCore}

\usepackage{tikz}
\usetikzlibrary{calc}
\usetikzlibrary{shapes.geometric, arrows}

\tikzstyle{question} = [rectangle,
minimum width=3cm,
minimum height=1cm,
text centered,
text width=3cm,
draw=black,
fill=orange!30]

\tikzstyle{rule} = [rectangle, rounded corners,
minimum width=1.5cm,
minimum height=1cm,
text centered,
draw=black,
fill=blue!20]

\tikzstyle{arrow} = [thick,->,>=stealth]

\usepackage{listings, silver, front-end}
\usepackage{color}
\usepackage{soul}
\usepackage{makecell}
\usepackage{multirow}
\usepackage{dirtree}
\usepackage{xspace}
\usepackage{float}

\usepackage{caption}
\usepackage{subcaption}

\usepackage{hyperref}
\usepackage{breakurl}

\usepackage{wrapfig}

\definecolor{light-gray}{gray}{0.9}
\definecolor{darkgreen}{rgb}{0.0,0.6,0.0}

\usepackage{tikz}

\usepackage[scaled=0.85]{beramono}
\usepackage[T1]{fontenc}

\newcommand\code[1]{{\texttt{#1}}}
\newcommand\secref[1]{\autoref{#1}}
\newcommand\appref[1]{\PackageError{vipermacros}{use appendixref instead of appref}{}}
\newcommand\defref[1]{\hyperref[#1]{Def.~\ref*{#1}}}
\newcommand\notref[1]{\hyperref[#1]{Not.~\ref*{#1}}}
\newcommand\lstref[1]{\hyperref[#1]{List.~\ref*{#1}}}
\newcommand\tabref[1]{\hyperref[#1]{Tab.~\ref*{#1}}}
\newcommand\figref[1]{\hyperref[#1]{Fig.~\ref*{#1}}}
\newcommand\thmref[1]{\hyperref[#1]{Thm.~\ref*{#1}}}
\newcommand\lemref[1]{\hyperref[#1]{Lemma~\ref*{#1}}}
\newcommand\propref[1]{\hyperref[#1]{Prop.~\ref*{#1}}}

\newcommand{\wrt}{{{w.r.t.\@}}}

\newcommand{\cf}{{\it{cf.~}}}

\usepackage[normalem]{ulem}

\ifcomments\else
\def\nocolour{ }
\fi

\newcommand{\soutifcolour}[1]{\ifdefined\nocolour{}\else{\sout{#1}}\fi}

\newcommand{\peter}[1]{\ifdefined\nocolour{#1}\else{\color{brickred}{#1}}\fi}

\newcommand{\thibault}[1]{\ifdefined\nocolour{{#1}}\else{\color{blue}{#1}}\fi}

\newcommand{\michael}[1]{\ifdefined\nocolour{{#1}}\else{\color{brown}{#1}}\fi}

\newcommand{\mout}[1]{\michael{\soutifcolour{#1}}}

\definecolor{gcolor}{rgb}{0.55, 0.71, 0.0}

\newcommand{\isabelle}{Isabelle/HOL}

\definecolor{darkred}{rgb}{0.55, 0.0, 0.0}

\usepackage{mathtools}
\usepackage{comment}

\usepackage{prooftree}

\makeatletter

\newskip \point

\def \premisespacing{\quad}

\def \RulePremisesNewlineMore[#1]#2.#3#4{\@ifnextchar\bgroup{\RulePremisesNewlineMore[#1]{#2}.{#3\premisespacing#4}}{\@ifnextchar.{\RulePremisesNewline[#1]{{\begin{array}{c}#2\\#3\premisespacing#4\end{array}}}}{\RuleMultiPremise[#1]{{\begin{array}{c}#2\\#3\end{array}}}{#4}}}}

\def \RulePremisesNewline[#1]#2.#3{\@ifnextchar\bgroup{\RulePremisesNewlineMore[#1]{#2}.{#3}}{\@ifnextchar.{\RulePremisesNewline[#1]{{\begin{array}{c}#2\\#3\end{array}}}}{\RuleMultiPremise[#1]{#2}{#3}}}}

\def \RuleMultiPremise[#1]#2#3{\@ifnextchar\bgroup{\RuleMultiPremise[#1]{#2\premisespacing#3}}{\@ifnextchar.{\RulePremisesNewline[#1]{#2\premisespacing#3}}{\prooftree #2\justifies#3 \using{#1}\endprooftree}}}

\def \RuleWithName[#1]#2{\@ifnextchar\bgroup {\RuleMultiPremise[#1]{#2}}{\@ifnextchar.{\RulePremisesNewline[#1]{#2}}{\prooftree \justifies #2 \using{#1} \endprooftree}}}

\def \RuleWithInfo[#1]{\@ifnextchar[{\RuleWithNameAndCondition[#1]}{\RuleWithName[(#1)]}}

\def \RuleWithNameAndCondition[#1][#2]{\RuleWithName[(#1)^{#2}]}

\def \Inf{\proofrulebaseline=2ex \abovedisplayskip12\point\belowdisplayskip12\point \abovedisplayshortskip8\point\belowdisplayshortskip8\point \@ifnextchar[{\RuleWithInfo}{\RuleWithName[ ]}}

\makeatother

\usepackage{amsthm}

\newtheorem{lemma}{Lemma}

\usepackage{xcolor}

\newcommand{\cseq}{ \mathbf{;} \; }

\usepackage{esvect}

\newcommand{\powerset}[1]{\mathbb{P}(#1)}

\newcommand{\namerule}[1]{\textit{#1}}

\renewcommand*{\proofname}{Proof sketch}

\newcommand{\freevars}[1]{\mathit{fv}(#1)}
\newcommand{\modvars}[1]{\mathit{mod}(#1)}

\usepackage{pifont}

\definecolor{royalblue}{rgb}{0.25, 0.41, 0.88}

\definecolor{somebrown}{rgb}{0.8, 0.35, 0.1}

\newcommand*{\fpointsto}[3]{\ensuremath{{#1}\overset{\scriptscriptstyle{#3}}{\mapsto}{#2}}}
\newcommand*{\pointsto}[2]{\ensuremath{{#1} \mapsto {#2}}}

\newcommand{\sepkeyword}[1]{{\texttt{\bfseries #1}}}

\newcommand{\facc}[2]{\sepkeyword{acc}(#1, #2)}
\newcommand{\wacc}[1]{\sepkeyword{acc}(#1,\_)} 
\newcommand{\acc}[1]{\sepkeyword{acc}(#1)}

\newcommand{\mvcode}[1]{\text{\vcode{#1}}}
\newcommand{\mfecode}[1]{\text{\fecode{#1}}}

\newcommand{\redStmt}[3]{\ensuremath{\langle #1, #2 \rangle \rightarrow_\Delta #3}}

\newcommand{\axiomSem}[3]{\ensuremath{\Delta \vdash [#1] \; #2 \; [#3]}}
\newcommand{\opName}[1]{\ruleref{#1Op}}
\newcommand{\opNameFull}[1]{\ruleref{#1OpFull}}
\newcommand{\axName}[1]{\ruleref{#1Ax}}
\newcommand{\axNameFull}[1]{\ruleref{#1AxFull}}
\newcommand{\CSL}[3]{\ensuremath{\Delta \vdash_{\mathit{CSL}} [#1] \; #2 \; [#3]}}

\newcommand{\anySL}[3]{\ensuremath{\Delta \vdash_{\mathcal{L}} [#1] \; #2 \; [#3]}}

\newcommand{\relStable}[2]{\textsf{frames}(#1, #2)}
\newcommand{\stable}[1]{\textsf{stable}(#1)}
\newcommand{\stabilize}[1]{\textsf{stabilize}(#1)}
\newcommand{\stableNa}{\textsf{stable}}
\newcommand{\stabilizeNa}{\textsf{stabilize}}

\newcommand{\selfFraming}[1]{\textsf{selfFraming}(#1)}
\newcommand{\framedBy}[2]{\textsf{frames}(#1, #2)}

\newcommand{\substitution}[3]{#1[#2/#3]}
\newcommand{\core}[1]{|#1|}

\newcommand{\parcomp}{\ensuremath{\mid \mid}}
\newcommand{\translate}[1]{ \llbracket #1 \rrbracket }

\newcommand{\IDFAlgebra}{IDF algebra}
\newcommand{\Stabilize}[1]{\ensuremath{{\ll}#1{\gg}}}
\newcommand{\fst}[1]{\ensuremath{#1.1}}
\newcommand{\snd}[1]{\ensuremath{#1.2}}

\newcommand{\IVLcompile}[1]{{\downarrow}\ \!{#1}}
\newcommand{\IVLstate}{\omega}

\newcommand{\assertion}{A}
\newcommand{\stmt}{C}

\newcommand{\sexec}{\textsf{sexec}}
\newcommand{\sproduce}{\textsf{sproduce}}
\newcommand{\sconsume}{\textsf{sconsume}}
\newcommand{\sstablize}{\textsf{scleanup}}
\newcommand{\sexp}{\textsf{sexp}}

\newcommand{\THSem}{VCGSem} 
\newcommand{\THSemShort}{VCG} 
\newcommand{\THSemRed}[4]{\langle #2, #3 \rangle \rightarrow_{\mathsf{\THSemShort{}}} #4}
\newcommand{\THSemStmt}{\stmt}

\newcommand{\THSemCtx}{\Gamma}
\newcommand{\THSemState}{\sigma_{t}}
\newcommand{\THSemStateB}{\sigma'_{t}}
\newcommand{\THSemRes}{r}
\newcommand{\THSemFail}{\textsf{F}}
\newcommand{\THSemMagic}{\textsf{M}}
\newcommand{\THSemNormal}[1]{\textsf{N}(#1)}

\newcommand{\symrel}[2]{#1 \sim_{\mathsf{sym}} #2}
\newcommand{\vcgrel}[2]{#1 \sim_{\mathsf{\THSemShort{}}} #2}

\newcommand{\Inhale}[1]{\vinhale{#1}}
\newcommand{\Exhale}[1]{\vexhale{#1}}
\newcommand{\SynIf}[3]{\textsf{if } #1 \textsf{ then } #2 \textsf{ else } #3}
\newcommand{\symexpr}{t}
\newcommand{\SymExpr}{\mathsf{SymExpr}}
\newcommand{\SCond}{K}
\newcommand{\Variable}{\mathsf{Var}}
\newcommand{\SState}{\sigma}
\newcommand{\SymState}{\mathsf{SymState}}
\newcommand{\symstore}{\mathsf{store}}

\newcommand{\sympc}{\mathsf{pc}}
\newcommand{\symheap}{\mathsf{heap}}
\newcommand{\symproves}[2]{#1.\sympc \vdash #2}
\newcommand{\ListTyp}[1]{\mathsf{List}(#1)}
\newcommand{\Chunk}{\mathsf{Chunk}}
\newcommand{\chunk}{c}
\newcommand{\symaddcond}[2]{\textsf{pc\_add}(#1, #2)}
\newcommand{\symaddchunkH}{\textsf{chunk\_add}}
\newcommand{\symaddchunk}[2]{\symaddchunkH~#1~#2}
\newcommand{\symconsolidateH}{\mathsf{consolidate}}
\newcommand{\symconsolidate}[1]{\symconsolidateH~#1}
\newcommand{\symextractH}{\textsf{extract}}
\newcommand{\symextract}[4]{\symextractH~#1~#2~#3~#4}
\newcommand{\symvar}{x}
\newcommand{\literal}{l}
\newcommand{\unoperation}{\odot}
\newcommand{\binoperation}{\oplus}
\newcommand{\symfresh}{\mathsf{fresh}}
\newcommand{\chunkfield}{\mathsf{field}}
\newcommand{\chunkrecv}{\mathsf{recv}}
\newcommand{\chunkperm}{\mathsf{perm}}
\newcommand{\chunkval}{\mathsf{val}}
\newcommand{\fieldname}{f}
\newcommand{\FieldName}{\mathsf{FieldName}}
\newcommand{\accmath}[3]{\facc{#1.#2}{#3}}
\newcommand{\wildcardmath}{\_}

\begin{document}

\ifextended
\title{Formal Foundations for Translational Separation Logic Verifiers (extended version)}
\fi


\author{Thibault Dardinier}
\orcid{0000-0003-2719-4856}
\affiliation{%
  \institution{ETH Zurich}
  \department{Department of Computer Science}
  \city{Zurich}
  \country{Switzerland}
}
\email{thibault.dardinier@inf.ethz.ch}

\author{Michael Sammler}
\orcid{0000-0003-4591-743X}
\affiliation{%
  \institution{ETH Zurich}
  \department{Department of Computer Science}
  \city{Zurich}
  \country{Switzerland}
}
\email{michael.sammler@inf.ethz.ch}

\author{Gaurav Parthasarathy}
\orcid{0000-0002-1816-9256}
\affiliation{%
  \institution{ETH Zurich}
  \department{Department of Computer Science}
  \city{Zurich}
  \country{Switzerland}
}
\email{gaurav.parthasarathy@inf.ethz.ch}

\author{Alexander J. Summers}
\orcid{0000-0001-5554-9381}
\affiliation{%
  \institution{University of British Columbia}
  \city{Vancouver}
  \country{Canada}
}
\email{alex.summers@ubc.ca}

\author{Peter Müller}
\orcid{0000-0001-7001-2566}
\affiliation{%
  \institution{ETH Zurich}
  \department{Department of Computer Science}
  \city{Zurich}
  \country{Switzerland}
}
\email{peter.mueller@inf.ethz.ch}

\begin{abstract}
Program verification tools are often implemented as front-end translations of an input program into an intermediate verification language (IVL) such as Boogie, GIL, Viper, or Why3. The resulting IVL program is then verified using an existing back-end verifier. A soundness proof for such a \emph{translational verifier} needs to relate the input program and verification logic to the semantics of the IVL, which in turn needs to be connected with the verification logic implemented in the back-end verifiers. Performing such proofs is challenging due to the large semantic gap between the input and output programs and logics, especially for complex verification logics such as separation logic.

This paper presents a formal framework for reasoning about translational separation logic verifiers. At its center is a generic core IVL that captures the essence of different separation logics. We define its operational semantics and formally connect it to two different back-end verifiers, which use symbolic execution and verification condition generation, resp. Crucially, this semantics uses angelic non-determinism to enable the application of different proof search algorithms and heuristics in the back-end verifiers. An  axiomatic semantics for the core IVL simplifies reasoning about the front-end translation by performing essential proof steps once and for all in the equivalence proof with the operational semantics rather than for each concrete front-end translation.

We illustrate the usefulness of our formal framework by instantiating our core IVL with elements of Viper and connecting it to two Viper back-ends as well as a front-end for concurrent separation logic.
All our technical results have been formalized in Isabelle/HOL,
including the core IVL and its semantics, the semantics of two back-ends for a subset of Viper, and all proofs.

\end{abstract}



\maketitle

\section{Introduction}
\label{sec:introduction}

Many program verification tools are organized into a \emph{front-end}, which encodes an input program along with its specification and verification logic into an intermediate verification language (IVL), and a \emph{back-end}, which computes proof obligations from the IVL program and discharges them, for instance, using an SMT solver. Examples of such \emph{translational verifiers} include
Civl~\cite{KraglQ21} and Dafny~\cite{Dafny} based on the Boogie IVL~\cite{Boogie},
Creusot~\cite{Creusot} and Frama-C~\cite{FramaC} based on Why3~\cite{Why3}, Gillian for C and JavaScript~\cite{GillianII} and Rust~\cite{GillianRust} based on GIL~\cite{GillianI}, as well as
Prusti~\cite{Prusti} and VerCors~\cite{VerCors} based on Viper~\cite{Viper}.

Developing a program verifier on top of an IVL has major engineering benefits. Most importantly, back-end verifiers, which often contain complex proof search algorithms, sophisticated optimizations, and functionality to communicate with solvers and to report errors, can be re-used across different verifiers, which reduces the effort of developing a program verifier dramatically.

On the other hand, formal reasoning about translational verifiers, in particular, proving their soundness, is more difficult than for  verifiers developed by embedding a program logic in an interactive theorem prover (such as Bedrock~\cite{Chlipala11}, VST~\cite{VST}, and RefinedC~\cite{RefinedC}).
Proving that a translational verifier is sound requires (1)~a formal semantics of the IVL as well as proofs that connect the IVL program (2)~to the verification back-end and (3)~to the input program.
While these steps have been studied for IVLs based on standard first-order logic~\cite{ParthasarathyMuellerSummers21, Cohen24, Herms13}, they pose additional challenges for IVLs that natively support more-complex widely-used reasoning principles such as those of separation logic \peter{(SL)}~\cite{Reynolds2002} (and variations such as implicit dynamic frames (IDF)~\cite{SmansJP12}). We focus on these IVLs, which are commonly-used and especially useful for building verifiers for heap-manipulating and concurrent programs.

\paragraph{Challenge~1: Defining the semantics of the IVL\@}
Standard programming languages and the intermediate languages used in compilers come with a notion of execution that can naturally be captured by an operational semantics. In contrast, IVLs are typically not designed to be executable, but instead to capture a wide range of verification problems and algorithms for solving them.

To capture different verification \emph{problems}, IVLs contain features that enable the encoding of a diverse set of input programs (\eg by offering generic operations suitable for encoding different concurrency primitives), specifications (\eg by offering rich assertion languages), and verification logics (\eg by supporting concepts such as framing). An IVL semantics must reflect this generality.
For instance, separation logic-based IVLs provide complex primitives for manipulating separation logic resources, which can be used to encode separation logic rules into the IVL\@.
As a result, these primitives can be used to encode a large variety of input program features including procedure calls, loops, and concurrency.

To capture different verification \emph{algorithms}, an IVL semantics must not prescribe \emph{how} to construct a proof and should instead abstract over different algorithms. Back-ends should have the freedom to apply various techniques to compute proof obligations (\eg symbolic execution or verification condition generation), to resolve trade-offs between completeness and automation (\eg by over-approximating proof obligations), and to discharge proof obligations (\eg{} instantiating existentially quantified variables in different ways). 
For instance, existing algorithms have different performance characteristics for different classes of verification problems~\cite{ViperAlgorithms}; an IVL semantics should provide the freedom to choose the best one for the problem at hand.
In practice, capturing different verification algorithms is important for verifiers with multiple back-ends for the same language (\eg{} based on either symbolic execution or verification condition generation). However, even a single back-end may offer a variety of different algorithms, which are chosen based on heuristics or configured by the user (\eg via command-line options). 
Moreover, back-ends along with their verification algorithms apply different algorithms over time, as their developers optimize existing verification algorithms or add support for new verification algorithms.

\paragraph{Challenge~2: Connecting the IVL to back-ends}

Soundness requires that successful verification of an IVL program by a back-end verifier implies the correctness of the IVL program. Since a back-end verifier's algorithm ultimately decides the outcome of a verification run, a soundness proof needs to formally connect the concrete verification algorithm to the IVL's semantics.
In particular, this soundness proof needs to consider the proof search algorithms and optimizations performed by a concrete verification back-end and show that they produce correct results according to the IVL semantics.
However, different back-ends typically use a diverse range of strategies to (for example) represent the program state, unroll recursive definitions, choose existentially-quantified permission amounts, and select the footprints of magic wands~\cite{DardinierPWMS22}.

\paragraph{Challenge~3: Connecting the IVL to front-ends}

Soundness also requires that the correctness of the IVL program implies the correctness of the \emph{input} program with respect to its intended verification logic. Such soundness proofs are difficult due to the large semantic gap between input and IVL programs. The two programs may use different reasoning concepts and proof rules, which need to be connected by a soundness proof.
This gap is particularly large for typical encodings into IVLs based on separation logic, because the verification logic for the source of this translation is typically different from the one for the IVL program, e.g., one of the vast wealth of concurrent separation logics.
For instance, a parallel composition of two threads in the input program is typically encoded as \emph{three sequential} IVL programs: two for the parallel branches, each of which is verified using a separate specification provided by the user, and one for the enclosing code, which composes the two specifications to encode the behavior of the parallel composition overall.
Such a translation of front-end proof rules into multiple sequential verification problems is not obvious; a soundness proof must bridge this gap.

\paragraph{Prior work}

Several works formalize aspects of translational verifiers with IVLs based on separation logic, but none of them addresses all three challenges outlined above.
For Viper, \citet{THSem} build a proof-producing version of Viper's verification condition generation back-end, but do not attempt to connect it to front-end languages nor give a general semantics for Viper that would also capture Viper's symbolic execution back-end. Similarly, \citet{GradualViper} formalize a version (only) of Viper's \emph{symbolic execution} back-end; their focus is on adapting it to gradual verification.
\citet{FeatherweightVerifast} show the soundness of the symbolic execution of VeriFast~\cite{VeriFast} \wrt{} an input C program.\footnote{VeriFast itself is not an IVL, but must address similar challenges to IVLs based on separation logic since VeriFast's symbolic execution is used to justify multiple front-end languages (C, Java, Rust) using separation logic reasoning; its symbolic execution also has strong similarities with IVL back-ends.}
However, VeriFast has only a single (symbolic execution) back-end that is used as the basis for multiple front-end languages (C, Java, Rust) and thus the formalization does not abstract over different verification algorithms.

\citet{GillianTR} briefly describe a soundness framework for GIL~\cite{GillianII}, a parametric program representation used by the Gillian project. GIL needs to be instantiated with a state model, primitive assertions, and memory actions to obtain specific intermediate representations (essentially, multiple IVLs) useful for different verification projects (e.g., for JavaScript~\cite{GillianII} and Rust~\cite{GillianRust}).
However, each GIL instantiation also determines the back-end verification algorithm. As such, there is no common semantics that abstracts over different verification algorithms.

\paragraph{This work}

In this paper, we present a framework for formally justifying translational separation logic verifiers. At its center is a generic IVL, called \theIVL{}, that captures the essence of different IVLs based on separation logics. In particular, \theIVL{} can be instantiated with different statements, assertion languages, and  separation algebras; we use a generalized notion of separation algebra that allows us to also model the implicit dynamic frames logic used in Viper.

To address Challenge~1 above, we define the semantics of \theIVL{} (and correspondingly, each of its instantiations) using \emph{dual (\ie demonic and  angelic) non-determinism}. Demonic non-determinism is a standard technique to verify properties for all inputs, thread schedules, etc. Our novel insight is to complement it with angelic non-determinism to abstract over the different proof search algorithms employed by back-ends. Intuitively, the IVL program is correct if \emph{any} of these algorithms succeeds, which is an angelic behavior.

To address Challenge~2, we define an operational semantics for \theIVL{}, which incorporates these notions of dual non-determinism and, like \theIVL{} itself, is parametric in the separation algebra to support both separation logic and IDF\@. An \emph{operational} semantics facilitates proving a formal connection to the concrete verification algorithms used in back-ends. Separation logic verifiers typically perform symbolic execution, which is typically described operationally~\cite{BoerB21} and (as we show) can be connected to our operational semantics via a standard simulation proof. Similarly, an operational IVL semantics is well-suited for formalizing the connections to back-ends that encode IVL programs into a further, more basic IVL, such as Viper's verification condition generator, which encodes Viper programs into Boogie.

To address Challenge~3, we define an axiomatic semantics for \theIVL{} and prove its equivalence to our operational semantics. An \emph{axiomatic} semantics facilitates proving a formal connection to the program logic used on the front-end level because both deal with derivations, which are often structurally related due to the compositional nature of most IVL translations. In addition, we are able to prove some powerful generic results about idiomatic encoding patterns once-and-for-all, further minimizing the instantiation-specific gap that a formal soundness proof needs to bridge.

We illustrate the practical applicability of our formal framework by instantiating \theIVL{} with elements of Viper. We use the resulting operational semantics to prove the soundness of two verification back-ends: a formalization of the central features of Viper's symbolic execution back-end, and a pre-existing formalization of Viper's verification condition generator~\cite{THSem}. These proofs demonstrate, in particular, that our use of angelic non-determinism allows us to capture these two rather disparate (and representative) back-ends. At the other end, we prove soundness of a front-end  based on concurrent separation logic using our axiomatic semantics. These proofs demonstrate that our framework effectively closes the large semantic gap between front-ends and back-ends and enables formal reasoning about the entire chain.

\paragraph{Contributions and outline}

We make the following technical contributions:

\begin{itemize}
\item We present a formal framework for reasoning about translational separation logic verifiers, via a parametric language \theIVL, for which we define a novel operational semantics combining core separation-logic reasoning principles and dual non-determinism.
We define an alternative axiomatic semantics, and show its equivalence with our operational semantics.

\item We define a Viper instantiation of \theIVL. We formalize and prove the soundness of the core of Viper's symbolic execution back-end. Similarly, we show soundness of an existing formalization of Viper's back-end based on verification condition generation. These proofs illustrate how angelic non-determinism can abstract over these different algorithmic choices.

\item We formalize a front-end for a simple concurrent language to be verified with concurrent separation logic, as well as its standard encoding as employed in translational verifiers, and prove this encoding sound with respect to our axiomatic semantics for \theIVL.
\end{itemize}

\noindent
We give an overview of our key ideas in \autoref{sec:key-ideas}. We define the operational and axiomatic IVL semantics in \autoref{sec:semantics}. We discuss how to prove back-end soundness in \autoref{sec:backends} and front-end soundness in \autoref{sec:frontends}. We discuss related work in
\autoref{sec:related-work} and conclude in \autoref{sec:conclusion}.

All formalizations and proofs in this paper are mechanized in the Isabelle proof assistant~\cite{Isabelle}
and our mechanization is publicly available~\cite{artifact}.
\section{Key Ideas}
\label{sec:key-ideas}

\definecolor{yellowboxcolor}{HTML}{ffffb3}
\colorlet{yellowbox}{yellowboxcolor!80}

\newcommand{\Cpq}{\ensuremath{C_V}}
\newcommand{\technicalSize}{\scriptsize}
\newcommand{\thmSize}{\footnotesize}
\newcommand{\refSize}{\footnotesize}

\begin{figure*}
  \centering
  \tikzstyle{every node}=[font=\small]
  \tikzstyle{block}=[inner sep=0.2em, draw, align=center, preaction={fill,white}]
  \pgfdeclarelayer{background}
  \pgfdeclarelayer{foreground}
  \pgfsetlayers{background,main,foreground}
  \begin{tikzpicture}
    \small

    \node (ivl) [block, fill=yellowbox] {Operational\\semantics};
    \path(ivl.west)+(-1.2, 0) node[anchor=east] (slproof) [block, fill=yellowbox] {Axiomatic\\semantics};
    \path[draw, <-, thick] (slproof) -- (ivl) node[pos=0.5, below] (abssemproof) {\thmSize \thmref{thm:op-to-ax-simple}};
    \path[draw, <-, thick] (ivl) to [out=160,in=20] node[midway, above] (abssemcomplete) {\thmSize \thmref{thm:completeness}} (slproof);


    \path(slproof.north)+(0, +0.2) node [] {\refSize \secref{subsec:ax-semantics}};
    \path(ivl.north)+(0, +0.2) node [] {\refSize \secref{subsec:op-semantics}};


    \path(slproof.west)+(-1.4, +0.8) node[anchor=east] (csl) [block, fill=yellowbox] {CSL};
    \path[draw, <-, thick] (csl.east) -- (slproof) node[pos=0.35, above, sloped] (cslproof) {\thmSize \thmref{thm:fe-soundness}};

    \path(csl.west)+(-1.3, -0) node[anchor=east] (parimp) [block, fill=yellowbox] {ParImp\\semantics};
    \path[draw, <-, thick] (parimp) -- (csl.west) node[pos=0.5, above] (cslproof) {\thmSize \thmref{thm:adequacy}};

    \path(parimp.north)+(0, +0.2) node [] {\refSize \secref{subsec:csl}};
    \path(csl.north)+(0, +0.2) node [] {\refSize \secref{subsec:csl}};

    \path(slproof.west)+(-1.4, -0.8) node[anchor=east] (other) [block, minimum width=20, minimum height=15, fill=yellowbox] {...};
    \path[draw, <-, thick, dashed] (other.east) -- (slproof);
    \path(other.west)+(-1.3, -0) node[anchor=east] (othersem) [block, fill=yellowbox, minimum width=20, minimum height=15] {...};
    \path[draw, <-, thick, dashed] (othersem.east) -- (other.west);


    \path(ivl.east)+(1.4, 1) node[anchor=west] (symexec) [block, fill=yellowbox] {Symbolic execution};
    \path(ivl.east)+(1.3, 0) node[anchor=west] (verifcond) [block, fill=yellowbox] {\THSem{}};
    \path(ivl.east)+(2, -1) node[anchor=west] (otherbackend) [block, fill=yellowbox, minimum width=25, minimum height=10] {...};
    \path[draw, ->, thick] (symexec.west) -- (ivl) node[pos=0.35, above, sloped] (symexecproof) {\thmSize \thmref{thm:sexec-sound}};
    \path[draw, ->, thick] (verifcond.west) -- (ivl) node[pos=0.45, above, sloped] (thmsemproof) {\thmSize \thmref{thm:thsem-sound}};
    \path[draw, ->, thick, dashed] (otherbackend.west) -- (ivl);

    \path(verifcond)+(1.5, 0) node[anchor=west] (boogie) [block, fill=yellowbox] {VCG};
    \path[draw, ->, thick] (boogie.west) -- (verifcond.east); 

    \path(symexec.north)+(0, +0.2) node [] {\refSize \secref{sec:symbolic-execution}};
    \path(verifcond.north)+(0, +0.2) node [] {\refSize \secref{sec:verif-cond-gener}};


    \path(parimp.north)+(1.5, 0.9) node (frontends) [] {\textbf{Front-ends} (\secref{sec:frontends})};
    \path let \p1 = (symexec) in let \p2 = (frontends) in
      (\x1, \y2)
      node (backends) [] {\textbf{Back-ends} (\secref{sec:backends})};
    \path($(frontends)!0.5!(backends)$) node (coreIVL) [align=center] {\textbf{CoreIVL} (\secref{sec:semantics})};
    \path (coreIVL)+(0, -0.4) node (vipercore) [] {ViperCore instantiation (\secref{sec:viper-inst})};

    \begin{pgfonlayer}{background}
      \path (slproof.north west)+(-0.3,0.9) node (abssemnw) {};
      \path (ivl.south east)+(+0.3,-0.7) node (abssemse) {};
      \path[rounded corners, draw=black!50, dashed, fill=black!10]
          (abssemnw) rectangle (abssemse);
    \end{pgfonlayer}

\end{tikzpicture}
  \caption{Overview of our framework and its application to Viper.
  The yellow boxes represent components of our framework (such as semantics and logics), while the arrows show the theorems that connect them.   The dashed arrows and the unlabeled yellow boxes represent potential additional front-ends and back-ends that could be connected to ViperCore.  CSL stands for concurrent separation logic, and VCG for verification condition generation.
  VCG has been formally connected to VCGSem by \citet{THSem}.
}
  \label{fig:overview}
\end{figure*}

In this section, we present the key ideas behind our work.
Our framework and its instantiation to Viper is presented in \figref{fig:overview}.
At its center is \emph{\theIVL{}} (depicted by the grey area), a general core language for representing SL-based IVLs. This core language
bridges the substantial gap between proofs of high-level programs using custom verification logics (\eg concurrent separation logic~\cite{CSL} (CSL) in the figure) at the front-end level and verification algorithms for SL-based IVLs at the back-end level (\eg symbolic execution and verification condition generation).
\theIVL{} is \emph{parametric} in its state model and assertions, so that it can represent multiple variants of separation logic (\eg those on which VeriFast and GIL are based), including implicit dynamic frames (on which Viper is based).
In \figref{fig:overview}, \ViperInst{} represents the instantiation of these parameters for Viper.
We give two equivalent semantics to \theIVL{}:
An \emph{operational} semantics, which is designed to enable soundness proofs for diverse back-end verification algorithms (shown on the right of \figref{fig:overview}),
and an \emph{axiomatic} semantics, which can be used to prove front-end translations into \theIVL{} sound, by connecting this axiomatic semantics to the front-end separation logic (shown on the left).

The rest of this section is organized as follows.
\secref{subsec:core-ivl} introduces the general core language \emph{\theIVL{}} for representing SL-based IVLs.
\secref{subsec:running-example} illustrates how to check for the existence of a Concurrent Separation Logic front-end proof for a parallel program by encoding the verification problem into our sequential \theIVL{}, mimicking the approach of modern translational verifiers.
\secref{subsec:op-semantics-key} presents the formal operational semantics of \theIVL{}.
Finally, \secref{subsec:ax-semantics-key} presents an alternative equivalent \emph{axiomatic semantics} for \theIVL{}, and shows how it can be leveraged to prove a front-end translation sound.

\subsection{A Core Language for SL-Based IVLs}
\label{subsec:core-ivl}

In this section, we first motivate and then define a core language for SL-based IVLs, called \emph{\theIVL{}}, which captures central aspects of SL-based verifiers, such as Viper~\cite{Viper}, Gillian~\cite{GillianI, GillianII}, or VeriFast~\cite{VeriFast}.

\paragraph{Manipulating SL states via inhale and exhale}
At the core of these verifiers is the SL state they track throughout the verification, typically containing a heap (a mapping from heap locations to values) and SL resources (such as fractional permissions to heap locations).
This SL state is manipulated with two verification primitives: $\vinhale{A}$ (also called \emph{assume*} and \emph{produce}) and $\vexhale{A}$ (also called \emph{assert*} and \emph{consume}), where $A$ is a separation logic assertion.
$\vinhale{A}$ assumes the logical constraints in $A$ (e.g.,  constraints on integer values), and adds the resources (e.g., ownership of heap locations) specified by $A$ to the current state.
Dually, $\vexhale{A}$ asserts that the logical constraints in $A$ hold, and removes the resources specified by $A$ from the current state.
These two primitives can encode the verification conditions for a wide variety of program constructs.
For instance, a procedure call is encoded as exhaling the call's precondition (to check its logical constraints and transfer ownership of resources from caller to callee), followed by inhaling the postcondition (to assume logical constraints and gain resources back from the call).

\paragraph{Diversity of logics and their semantics}
While SL-based IVLs all employ some version of these two inhale and exhale primitives, their actual logics are surprisingly diverse in both core connectives and their semantics. GIL and VeriFast support different separation logics, while Viper uses implicit dynamic frames (IDF), a variation of separation logic that allows for heap-dependent expressions in assertions (e.g., separation logic's points-to predicate $\pointsto{e.f}{v}$ is expressed as $\acc{e.f} * e.f = v$ in IDF, in which the ownership of the heap location and a logical constraint on its value are expressed as two separate conjuncts)\footnote{This difference also affects the semantic models; separation logic is typically formalized using  partial heaps, whereas IDF typically uses a total heaps model~\cite{ParkinsonSummers12}.}.

IVLs also support different SL connectives: Viper supports iterated separating conjunctions~\cite{ViperQPs},
Viper and Gillian support magic wands~\cite{DardinierPWMS22,SchwerhoffS15},
Viper and VeriFast support fractional recursively-defined predicates~\cite{Boyland03,DardinierMS22}, and
VeriFast supports arbitrary existential quantification.

A standard approach for generic reasoning over large classes of separation logics is to build reasoning principles based on a \emph{separation algebra} (built over a partial commutative monoid)~\cite{Calcagno2007,Dockins2009}. We extend this classic concept to a novel notion of \emph{\IDFAlgebra{}}, which can model separation logics and IDF alike. In particular,  \IDFAlgebra{}s allow asserting knowledge about the value of heap locations $e.f$ without asserting ownership of the heap location itself.

\begin{figure}
  $$
  C \Coloneqq \;
  \vinhale{A} \mid \vexhale{A} \mid
  \vhavoc{x} \mid
  C; C \mid \vif{b}{C}{C} \mid \vassign{x}{e} \mid \vskipp \mid \vcustom{C'}
  $$
\caption{Syntax of statements in \theIVL{}.
$A$ is an assertion, $x$ a variable, $b$ a Boolean expression, $e$ an arbitrary expression.
Assertions and expressions are represented semantically as  sets of states and partial functions from states to values, respectively.
$C'$ represents custom statements and is a parameter of the language.
}
\label{fig:viper-language}
\end{figure}

\paragraph{Core Language}
The syntax of \theIVL{} is shown in \figref{fig:viper-language}.
To capture the diversity of assertions supported in existing SL-based IVLs, assertions $A$ in our core language are \emph{semantic},
\ie assertions are \emph{sets of states} (as opposed to fixing a syntax, and having the semantics for this syntax determine the set of states in which a syntactic assertion is true); states themselves are taken from any chosen \IDFAlgebra{}.
Similarly, expressions $e$ are semantically represented as partial functions from states to values.
Moreover, although we assume some core statements in our language, we allow these to be arbitrarily extended via a parameter for \emph{custom statements} $C'$, for instance, to add field assignments. The statements of our core language contain the key verification primitives \vinhaleNa{} and \vexhaleNa{} described above, as well as
\vhavocNa{}, which non-deterministically assigns a value to a variable.
Combined with conditional branching, \vinhaleNa{}, \vexhaleNa{}, and \vhavocNa{} allow us to encode many important statements, such as while loops, procedure calls, and even proof rules for parallel programs, as we show in the next subsection.

\subsection{Background: Translational Verification of a Parallel Program}
\label{subsec:running-example}

\begin{figure}
 \small
\begin{minipage}{0.4\textwidth}
\begin{frontend}[mathescape]
method main(p: Cell)
  // requires acc(p.v, _)
{
  q := alloc(0)


  // ${\color{darkgreen} \{P_l\}}$         ${\color{darkgreen}\{P_r\}}$
  q.v := p.v || tmp := p.v
  // ${\color{darkgreen} \{Q_l\}}$         ${\color{darkgreen}\{Q_r\}}$

  tmp := tmp + q.v

  free(q)

  assert tmp = p.v + p.v
}
\end{frontend}
\end{minipage}
\begin{minipage}{0.35\textwidth}
\begin{viper}[mathescape]
method main_ivl(p: Ref) {
  inhale acc(p.v, _)

  havoc q
  inhale acc(q.v) * q.v = 0

  exhale $P_l$ * $P_r$
  havoc tmp
  inhale $Q_l$ * $Q_r$

  tmp := tmp + q.v

  exhale acc(q.v)

  exhale tmp = p.v + p.v
}
\end{viper}
\end{minipage}
\begin{minipage}{0.25\textwidth}
\begin{viper}
method l(p,q:Ref){
  inhale $P_l$
  q.v := p.v
  exhale $Q_l$
}

method r(p,q:Ref){
  inhale $P_r$
  tmp := p.v
  exhale $Q_r$
}
\end{viper}
\end{minipage}
\caption{A simple parallel program (left),
annotated with a method precondition, as well as pre- and postconditions for the parallel branches, and its encoding into \theIVL{} (instantiated to model Viper),
consisting of a main IVL method (middle)
and two further methods (right) modeling the parallel branches (that is, the premises of CSL's parallel composition rule). We use the shorthands
$P_l \triangleq \wacc{p.v} * \acc{q.v}$,
$Q_l \triangleq \wacc{p.v} * \acc{q.v} * \code{p.v} = \code{q.v}$,
$P_r \triangleq \wacc{p.v}$, and
$Q_r \triangleq \wacc{p.v} * \code{tmp} = \code{p.v}$,
where the IDF assertion $\wacc{e}$ expresses non-zero permission to $e$ (corresponding to the SL assertion $\exists p, v \ldotp \fpointsto{e}{v}{p}$).
}
\label{fig:running-example}
\end{figure}

We use the parallel program on the left in \figref{fig:running-example} to illustrate how translational verification works, and the challenges that arise in formalizing this widely-used approach.
This program takes as input a Cell \code{p} (an object with a value field \code{v}), allocates a new Cell \code{q}, assigns
the value of \code{p.v} in parallel to \code{q.v} and to the variable \code{tmp}, then adds the value of \code{q.v} to \code{tmp},
deallocates \code{q}, and finally asserts that \code{tmp} is equal to \code{p.v + p.v}.
Our goal is to verify this program in Concurrent Separation Logic (CSL)~\cite{CSL}, that is, by encoding the program and the proof rules of CSL into \theIVL{}. In particular, we want to prove that the assertion on its last line holds.

Although the original CSL is presented via standard separation logic syntax, we use the syntax of IDF to annotate this example. The syntax \facc{\code{e.v}}{f} denotes \emph{fractional permission} (ownership)  of the heap location \code{e.v} (where $f = 1$ allows reading and writing, and a fraction $0 < f < 1$ allows reading) \cite{Boyland03}. The syntax $\wacc{\code{p.v}}$ (used as precondition in our example) denotes a so-called \emph{wildcard permission} (or \emph{wildcard} in short); it is shorthand for $\exists f > 0 \ldotp \facc{\code{p.v}}{f}$, which guarantees read access while abstracting the precise fraction.

Correctness of our example means proving a CSL triple $\CSL{\wacc{\code{p.v}}}{C}{\top}$,
where $C$ is the body of the method \fecode{main} in the front-end (left) program ($\top$ is the trivial postcondition). Instead of constructing a proof directly, a translational verifier maps this to an IVL program (shown as a \theIVL{} program to the middle and right of \figref{fig:running-example}) whose correctness implies the existence of a CSL proof for the original program.

\paragraph{Encoding the program into \theIVL{}}

Our encoding models each proof task of the CSL verification problem as a separate IVL method, whose statements reflect the individual proof steps \cite{LeinoM09}. The IVL methods \code{main\_ivl}, \code{l} and \code{r} are constructed such that the correctness of \emph{all three} implies the existence of a valid CSL proof for \code{main}.

The precondition $\wacc{\code{p.v}}$ of \code{main} is modeled by the first \vinhaleNa{} statement in \code{main\_ivl}, reflecting that the proof of the main method may rely on the resources and assumptions guaranteed by this precondition.
The allocation \fecode{q := alloc(0)} is then encoded via a  \vhavocNa{} and an \vinhaleNa{} statement to non-deterministically choose a memory location and obtain a full (\ie $1$) permission.
Dually, the deallocation \fecode{free(q)} after the parallel composition is encoded via an \vexhaleNa{} statement, which removes this (full) permission from the IVL state. Since permissions are non-duplicable (technically, affine) resources, this encoding guarantees that no permission can remain and so any attempt to later access this location would cause a verification failure.

To understand the encoding of a source-level parallel composition, we recall the CSL proof rule\footnote{We omit technical side-conditions from the original rule that restrict mutation of variables shared amongst threads; these are taken care of properly in real verifiers and our formalizations.}:
\begin{mathpar}
  \inferHL{Par}
  {\CSL{P_l}{C_l}{Q_l} \\ \CSL{P_r}{C_r}{Q_r}}
  {\CSL{P_l * P_r}{C_l \; || \; C_r}{Q_l * Q_r}}
\end{mathpar}
From the point of view of the \emph{outer thread} (forking and joining the parallel branches), the overall effect of the parallel composition can be seen as \emph{giving up} the separating conjunction ${P_l * P_r}$ of the preconditions of the parallel branches, and obtaining the corresponding postconditions ${Q_l * Q_r}$ before resuming any remaining code\footnote{We assume (as is common for modular verifiers) that each thread's specification is explicitly annotated, as in \code{main}.}. This exchange of assertions across the triple in the conclusion of the rule (as well as the intervening modification of \code{tmp}) is modeled in the IVL program by the sequence $\vexhale{P_l*P_r};\vhavoc{\code{tmp}};\vinhale{Q_l*Q_r}$.

The premises of the parallel rule are checked by verifying two extra methods \code{l} and \code{r}, whose pre- and postconditions correspond to the Hoare triples from the rule premises directly.
The encoded bodies of \code{l} and \code{r} follow the standard pattern: an \vinhaleNa{} of their preconditions (which can be seen as the other ``half'' of the transfer from the outer thread, modeled by $\vexhale{P_l*P_r}$), the translation of their source implementations, and finally an \vexhaleNa{} of their postconditions.

If running a back-end verifier for the IVL on the three encoded methods succeeds, we have demonstrated that a CSL proof for the original program exists---provided that the translational verification is sound. Soundness depends on a non-trivial translation, the subtle semantics of an IVL, and the algorithms employed by back-end verifiers. In the rest of this section, we explain our formal framework for establishing the soundness of translational verifiers.

\subsection{Operational Semantics and Back-End Verifiers}
\label{subsec:op-semantics-key}

To make formal claims about an IVL program,  we need a formal  semantics and notion of correctness for the IVL itself. As explained in the introduction, an operational semantics facilitates a formal connection to various back-end algorithms, which typically have an operational flavor. Since our semantics needs to capture verification algorithms that make heavy use of (demonic) non-determinism (to model concurrency, allocation, or abstract modularly over the precise behavior of program elements), our operational semantics embraces such non-determinism. Moreover, to account for the diversity of the verification algorithms used in back-ends, our semantics also incorporates the dual notion of \emph{angelic non-determinism}.

Consider verifying the statement
$\vexhale{\acc{\code{a.v}} \lor \acc{\code{b.v}}}$, which requires giving up (full) permission to \emph{either} \code{a.v} or \code{b.v}; if the original state holds both permissions, either choice avoids a failure here, but results in different successor states, and so might affect whether subsequent statements verify successfully. Such algorithmic choices occur for other IVL constructs, such as for choosing the values of existentials (including the amount of permission for a wildcard permission), or determining the footprints of magic wands. Our operational semantics makes \emph{all algorithmic choices possible} and defines a program as correct if \emph{any} such choice avoids failure.

\begin{figure}
\footnotesize
\begin{subfigure}[t]{\textwidth}
\begin{mathpar}
\inferH{InhaleOp}{}
{\redStmt{\vinhale{A}}{\omega}{ \{ \omega \} * A }}
\and
%
\inferH{ExhaleOp}
{\omega = \omega' \oplus \omega_A \\
\omega_A \in A
}
{\redStmt{\vexhale{A}}{\omega}{ \{ \omega' \}}}
\and
%
%
%
\inferH{SeqOp}
{\redStmt{C_1}{\omega}{S_1} \\
\forall \omega_1 \in S_1 \ldotp \redStmt{\thibault{C_2}}{\thibault{\omega_1}}{\thibault{\mathcal{S}}(\omega_1)} }
{\redStmt{C_1;C_2}{\omega}{ \cup_{\omega_1 \in S_1} \thibault{\mathcal{S}}(\omega_1)}}
\end{mathpar}
\caption{Selected operational semantics rules.}
\label{fig:op-semantics-small}
\end{subfigure}
\begin{subfigure}[t]{\textwidth}
\begin{mathpar}
\inferH{InhaleAx}{}
{\axiomSem{P}{\vinhale{A}}{P * A }}
\and
%
\inferH{ExhaleAx}
{P \models Q * A
}
{\axiomSem{P}{\vexhale{A}}{Q}}
\and
%
%
%
\inferH{SeqAx}
{\axiomSem{P}{C_1}{R} \\
\axiomSem{R}{C_2}{Q}}
{\axiomSem{P}{C_1;C_2}{Q}}
\end{mathpar}
\caption{Selected axiomatic semantic rules.}
\label{fig:ax-semantics-small}
\end{subfigure}
\caption{Selected simplified operational and axiomatic semantic rules.}
\label{fig:semantics-small}
\end{figure}

\paragraph{Operational semantics.}
To capture the dual non-determinism, we define our operational semantics as a multi-relation~\cite{BinaryMultirelations, Melocoton}
$$
\redStmt{C}{\omega}{S}
$$
where $C$ is an IVL statement, $\omega$ an initial state, $S$ a set of final states, and $\Delta$ a type context (mapping for example variables to types, \ie to sets of values).
The set $S$ captures the \emph{demonic choices}, \ie contains the resulting state for each possible demonic choice. On the other hand, \emph{angelic choices} are reflected by \emph{different result sets} derivable in our semantics. Returning to our previous example, if $\omega$ is a state with full permission to both \code{a.v} and \code{b.v}, our semantics allows for both transitions
$\redStmt{\vexhale{\acc{\code{a.v}} \lor \acc{\code{b.v}}}}{\omega}{\{ \omega_{-a} \}}$ and $\redStmt{\vexhale{\acc{\code{a.v}} \lor \acc{\code{b.v}}}}{\omega}{\{ \omega_{-b} \}}$ (where $\omega_{-a}$ and $\omega_{-b}$ are identical to $\omega$ but with the permission to \code{a.v} resp. \code{b.v} removed).

A successful verification by a back-end is represented by an execution in our operational semantics, leading to the following definition of correctness of a \theIVL{} statement:
\begin{definition}\label{def:correct}
  A \theIVL{} statement $C$ is \emph{correct} for a well-formed initial state $\omega$ iff $C$ executes successfully in $\omega$,
  \ie $\exists S \ldotp \redStmt{C}{\omega}{S}$.
$C$ is \emph{valid} iff it is correct for all well-formed initial states.
\end{definition}

\figref{fig:op-semantics-small} shows simplified rules for the operational semantics of $\vinhale{A}$, $\vexhale{A}$, and sequential composition. The (non-simplified) rules for all statements are shown in \secref{sec:semantics}.
Inhaling A in state $\omega$ leads to the set of all possible combinations $\omega \oplus \omega_A$ for $\omega_A \in A$, capturing the demonic non-determinism of \vcode{inhale}: All possible states satisfying $A$ must be considered in the rest of the program.
Dually, the rule \opName{Exhale} allows \emph{any choice of state} $\omega_A$ satisfying $A$ (that is, uses angelic non-determinism), and to remove it from $\omega$.
In our previous example, $\omega$ can be decomposed into
$\omega = \omega_{-a} \oplus \omega_a$ or $\omega = \omega_{-b} \oplus \omega_b$,
where $\omega_a$ and $\omega_b$ respectively contain the permission to \code{a.v} and \code{b.v} (and thus $\omega_a$ and $\omega_b$
both satisfy the exhaled assertion $\acc{\code{a.v}} \lor \acc{\code{b.v}}$).
The rule \opName{Seq} for sequential composition is more involved, since it needs to deal with the dual non-determinism:
It requires a single function $\mathcal{S}$ that maps every state $\omega_1$ from $S_1$ (the set of states obtained after executing $C_1$ in $\omega$)
to a set of states $\mathcal{S}(\omega_1)$ that can be reached by executing $C_2$ in $\omega_1$.
The choice of the function $\mathcal{S}$ captures the angelic choices in $C_2$.

\paragraph{Connection to back-end verifiers}
To show that this operational semantics for \theIVL{} is indeed suitable to capture different verification algorithms, we connect it to formalizations of the two main back-ends used by Viper.
First, we formalize a version of Viper's symbolic execution back-end~\cite{MalteThesis} in \isabelle{} and prove it sound against the operational semantics of \theIVL{}.
Second, we connect the formalization of Viper's verification condition generation back-end by \citet{THSem} to \theIVL{} by constructing a \theIVL{} execution from a successful verification by their back-end.
The soundness proofs of these back-ends are described in \secref{sec:backends}.
There we will also see that the angelic choice described earlier in this section is crucial for enabling these proofs since the two back-ends use different heuristics, in particular around exhaling wildcard permissions.

\subsection{Axiomatic Semantics}
\label{subsec:ax-semantics-key}

The previously-introduced definition of correctness (\defref{def:correct}) based on the operational semantics is well-suited to connect to back-end verifiers. However, connecting it to front-end programs, and especially logics such as CSL in our example from \figref{fig:running-example}, requires substantial effort due to the large semantic gap between the operational IVL semantics and the front-end logic.
The IVL semantics presented previously is \emph{operational}, describes the execution from a \emph{single state}, and exposes low-level details (such as handling the dual non-determinism in the rule \opName{Seq}).
In contrast, the program logic is \emph{axiomatic}, describes the behavior of \emph{sets of states} (via assertions), and is more high-level (\eg it uses an intermediate assertion in the rule \axName{Seq} instead of the semantic function $\mathcal{S}$).
To bridge this gap,
we present an alternative (and, as we later prove, equivalent) axiomatic semantics for \theIVL{},
which is closer to the separation logics typically used for front-end programs and, thus, simplify the proof that a front-end translation is sound.

Our axiomatic semantics uses judgments of the form
$$\axiomSem{P}{C}{Q}$$
where $P$ and $Q$ are semantic assertions (sets of states),
$C$ is an IVL statement, and $\Delta$ is a type context.
Intuitively, this triple expresses that $C$ can be executed successfully in any state from $P$ (with the right angelic choices),
and $Q$ is (precisely) the set of all states reached by these executions.
Formally, we want the following \emph{soundness} property (we will present the completeness theorem in \autoref{sec:semantics}):
\begin{theorem}[Operational-to-Axiomatic Soundness.]\label{thm:op-to-ax-simple}
If the \theIVL{} statement $C$ is well-typed and valid (\defref{def:correct})
  then there exists a set of states $B$ such that $\axiomSem{\top}{C}{B}$ holds.
\end{theorem}
Note that, in contrast to when one defines a proof system for a pre-existing operational semantics, the desired implication here is from operational to axiomatic semantics; this is due to the connection we are aiming for from back-end algorithms (defined operationally) to front-end proofs.

The rules for the axiomatic semantics of $\vinhale{A}$, $\vexhale{A}$, and sequential composition are shown in \figref{fig:ax-semantics-small}.
The rule \axName{Inhale} for $\vinhale{A}$ corresponds to the operational rule \opName{Inhale}, where $\omega$ has been lifted to the set of states $P$ (since $P * A = \bigcup_{\omega \in P} (\{ \omega \} * A$)).
The rule \axName{Exhale} for $\vexhale{A}$ is more involved, as it first requires weakening the set of initial states $P$ to $Q * A$.
Weakening is in general necessary to disentangle the states in $Q$ and $A$: For example, to exhale $\acc{\code{a.v}}$ from a precondition $\acc{\code{a.v}} * \acc{\code{b.v}} * \code{a.v} = \code{b.v}$, we have to first drop the equality $\code{a.v} = \code{b.v}$ because otherwise the resulting postcondition would refer to a memory location that is no longer owned.
Moreover, similarly to how Hoare logic hides the induction necessary to reason about unbounded while loops behind a loop invariant, our axiomatic semantics
hides the dual non-determinism of the operational semantics behind high-level connectives such as the separating conjunction.
Intuitively, in the rule \axName{Exhale}, the angelic choice is hidden in the choice of $Q$ and the split of every state in $P$ into a state in $Q$ and a state in $A$.
In our previous example $\vexhale{\acc{\code{a.v}} \lor \acc{\code{b.v}}}$,
we could choose $Q$ to be either $\acc{\code{a.v}}$ or $\acc{\code{b.v}}$,
\ie we could derive both $\axiomSem{ \acc{\code{a.v}} * \acc{\code{b.v}} }{\vexhale{\acc{\code{a.v}} \lor \acc{\code{b.v}}}}{\acc{\code{a.v}}}$ and $\axiomSem{ \acc{\code{a.v}} * \acc{\code{b.v}} }{\vexhale{\acc{\code{a.v}} \lor \acc{\code{b.v}}}}{\acc{\code{b.v}}}$.

Finally, the rule \axName{Seq} for sequential composition illustrates how the axiomatic semantics abstracts over the low-level details of the dual non-determinism in the operational semantics, such as the existence of the semantic function $\mathcal{S}$ in rule \opName{Seq}. Instead, the axiomatic rule \axName{Seq} uses an intermediate assertion $R$; its relation to $\mathcal{S}$ is proved once and for all in the soundness proof and, thus, does not have to be proved for each front-end.

Crucially, we have designed the axiomatic semantics such that it contains \emph{exactly one rule} per statement.
In particular, it contains no structural rules such as a frame rule or a consequence rule, which are not necessary in our setting.
This allows us to deconstruct an axiomatic semantic derivation into smaller blocks, to then reconstruct a proof in the front-end logic.
For example, one can derive from $\axiomSem{P}{C_1; C_2}{Q}$ the existence of some assertion $R$ such that $\axiomSem{P}{C_1}{R}$ and $\axiomSem{R}{C_2}{Q}$ hold.
Using this axiomatic semantics, we can now easily connect the correctness of the IVL program to the correctness of the front-end program, as we explain next.

\begin{figure}
\small
\begin{mathpar}
\inferhref{Frame}{FrameKey}
{\CSL{P}{C}{Q} \\
\freevars{F} \cap \modvars{C} = \emptyset}
{\CSL{P * F}{C}{Q * F}}
\and
\inferhref{Par}{ParKey}
{\CSL{P_l}{C_l}{Q_l} \\
\CSL{P_r}{C_r}{Q_r} \\
\ldots}
{\CSL{P_l * P_r}{C_l \; || \; C_r}{Q_l * Q_r}}
\and
\inferhref{Seq}{SeqKey}
{\CSL{P}{C_1}{R} \\
\CSL{R}{C_2}{Q}}
{\CSL{P}{C_1;C_2}{Q}}
\and
\inferhref{Cons}{ConsKey}
{\CSL{P'}{C}{Q'} \\
P \models P' \\
Q' \models Q}
{\CSL{P}{C}{Q}}
\and
\inferhref{Alloc}{AllocKey}
{r \notin \freevars{e}}
{\CSL{\top}{\code{r} := \fenew{e}}{ \acc{\code{r.v}} * r.v = e }}
\and
\inferhref{Free}{FreeKey}
{}
{\CSL{\acc{\code{q.v}}}{\fefree{\code{q}}}{ \top }}
\end{mathpar}
\caption{Selected CSL rules.
In the rule \ruleref{FrameKey},
$\freevars{F}$ and $\modvars{C}$ denote the set of variables free in $F$ and the set of variables potentially modified by $C$, respectively.}
\label{fig:csl-rules}
\end{figure}

\paragraph{Connecting to front-end programs and logics}
Let us now see how the axiomatic semantics enables us to construct a CSL proof for the front-end program from \figref{fig:running-example}.
Concretely, we build a CSL proof of the triple $\CSL{\wacc{\code{p.v}}}{C}{\top}$, where $C$ corresponds to the body of the method \fecode{main}.
To do this, we use the CSL rules shown in \figref{fig:csl-rules} and the \theIVL{} triples $\axiomSem{\top}{C}{B}$ for the methods \code{l}, \code{r}, and \code{main\_ivl} that we obtain from \thmref{thm:op-to-ax-simple}.

The first step of proving the CSL triple for \fecode{main} is to pair each statement in \fecode{main} with the corresponding code in \code{main\_ivl}.
For this, we use CSL's \ruleref{SeqKey}
rule and (the inversion of) \axName{Seq} to split the proofs for \code{main} and \code{main\_ivl} into smaller parts:
{
  \small
\begin{align*}
  &&& \axiomSem{\top}{\mvcode{inhale acc(p.v, \_)}}{A_0} \\
  &\CSL{A_0}{\mfecode{q := alloc(0)}}{A_1} && \axiomSem{A_0}{\mvcode{havoc q;} \mvcode{inhale acc(q.v) * q.v = 0}}{A_1} \\
  &\CSL{A_1}{\mfecode{q.v := p.v \|\| tmp := p.v}}{A_2} && \axiomSem{A_1}{\vexhale{P_l{*}P_r};\kern1pt\vhavoc{\code{tmp}};\kern1pt\vinhale{Q_l{*}Q_r}}{A_2} \\
  &\CSL{A_2}{\mfecode{tmp := tmp + q.v}}{A_3} && \axiomSem{A_2}{\mvcode{tmp := tmp + q.v}}{A_3} \\
  &\CSL{A_3}{\mfecode{free(q)}}{A_4} && \axiomSem{A_3}{\mvcode{exhale acc(q.v)}}{A_4} \\
  &\CSL{A_4}{\mfecode{assert tmp = p.v + p.v}}{B} && \axiomSem{A_4}{\mvcode{exhale tmp = p.v + p.v}}{B}
\end{align*}
}
\noindent
Note how deconstructing the applications of \axName{Seq} in the proof of \code{main\_ivl}  gives us intermediate assertions $A_{0-4}$, which we use to instantiate the intermediate assertion $R$ in \ruleref{SeqKey}.%
\footnote{Note that the CSL we use in this paper has the same state model as the IVL, and thus the IVL assertions do not need to be converted to CSL assertions.
Our axiomatic semantics can also be used to reconstruct proofs in program logics with different state models, but this goes beyond the scope of this paper.}
Matching statements of the front-end program to segments of the \theIVL{} program is straightforward since the front-end translation is typically defined statement by statement.

After decomposing the sequential compositions, we justify the CSL triple for each primitive front-end statement from the corresponding \theIVL{} triple.
For some statements like \mfecode{tmp := tmp + q.v}, this is trivial as the triples (and corresponding logic rules) match.
Let us now focus on the most interesting cases: \mfecode{q := alloc(0)}, \mfecode{q.v := p.v \|\| tmp := p.v}, and \mfecode{free(q)}.

\paragraph{The exhale-havoc-inhale pattern}
To derive the CSL triples for these statements, we observe that their encoding follows a pattern:
The \theIVL{} code first exhales the precondition $P$ of the CSL rule (omitted if $P=\top$), then havocs the variables modified by the statement (\code{q} for \mfecode{q := alloc(0)} and \code{tmp} for \mfecode{q.v := p.v \|\| tmp := p.v}), and finally inhales the postcondition $Q$ of the CSL rules (omitted if $Q=\top$), leading to the pattern
$\vexhale{P}; \vhavoc{x_1}; \ldots; \vhavoc{x_n}; \vinhale{Q}$.
To handle this general pattern, we can use the following lemma, which holds for any separation logic $\mathcal{L}$ with a consequence rule and a frame rule (see \autoref{sec:frontends} for the proof):

\begin{lemma}[Exhale-havoc-inhale]\label{lem:exhale-inhale-key}
For any separation logic $\mathcal{L}$ that has a frame rule and a consequence rule, if
$\anySL{P}{C}{Q}$ holds and
$\axiomSem{A}{\vexhale{P}; \vhavoc{x_1}; \ldots; \vhavoc{x_n}; \vinhale{Q}}{B}$ holds,
  where $\{ x_1, \ldots, x_n \} = \modvars{C}$,
 then
  $\anySL{A}{C}{B}$ holds.
\end{lemma}

Intuitively, this lemma shows that a \theIVL{} triple for the exhale-havoc-inhale pattern allows us to obtain the corresponding CSL triple.
In the case of \mfecode{q := alloc(0)}, this lets us lift \ruleref{AllocKey} to the precondition $A_0$ and postcondition $A_1$, giving us exactly the triple we need.
To justify the triple for \mfecode{q.v := p.v \|\| tmp := p.v}, we need to establish the premises of the rule \ruleref{Par},
$\CSL{P_l}{\mfecode{q.v := p.v } }{Q_l}$ and $\CSL{P_r}{ \mfecode{tmp := p.v} }{Q_r}$,
which can be derived from the correctness of the methods \code{l} and \code{r}
using a lemma similar to \lemref{lem:exhale-inhale-key},
as we formally show in \secref{sec:frontends}.

\paragraph{Summary}
We have now seen how to justify the translational verification of the program from \figref{fig:running-example} in CSL in three steps. First, we showed that the successful verification of its \theIVL{} encoding in a back-end implies that the \theIVL{} program is valid. Second, we used the soundness theorem for the axiomatic IVL semantics to derive judgments in the axiomatic semantics. Third, we use those judgments to prove the desired CSL triple.
Each of these steps is well-suited for its task: The operational semantics allows us to connect to the back-end verifiers, while the axiomatic semantics facilitates the reconstruction of the front-end logic proof---both linked by \thmref{thm:op-to-ax-simple}.
\section{Semantics}
\label{sec:semantics}

In this section, we present an operational and an axiomatic semantics for the \theIVL{} language defined in \figref{fig:viper-language}.
We first define in \secref{subsec:idf-algebra} an \IDFAlgebra{} that captures both separation logic and implicit dynamic frames state models.
We then formalize the operational semantics of \theIVL{} in \secref{subsec:op-semantics} and define its axiomatic semantics and prove their equivalence in \secref{subsec:ax-semantics}. We instantiate \theIVL{} for key features of Viper in \secref{sec:viper-inst}.

\subsection{An Algebra for Separation Logic and Implicit Dynamic Frames}
\label{subsec:idf-algebra}

A standard way to capture different separation logic state models is to use a \emph{separation algebra}~\cite{Calcagno2007,Dockins2009}, \ie a partial commutative monoid $(\Sigma, \oplus)$, where $\Sigma$ is the set of all states, and $\oplus$ is a partial, commutative, and associative binary operator, used to combine states (\eg via the separating conjunction operator $*$).
In SL, assertions about values of heap locations must also assert ownership of those heap locations.
In particular, asserting that a heap location \code{x.f} has a value $5$ requires using the points-to predicate  $\pointsto{\code{x.f}}{5}$), which also expresses ownership of the location \code{x.f}.
This requirement is embedded in the SL state model.
For example, a typical SL state with a heap and fractional permissions (ignoring local variables for now) is
$\Sigma_{\mathit{SL}} \triangleq \left( L \rightharpoonup (V \times (0, 1]) \right)$, \ie a partial function from a set $L$ of heap locations to pairs of values from a set $V$ and positive fractional permissions.
That is, any value for a heap location is associated with a strictly positive permission.

In contrast, in implicit dynamic frames, an assertion may constrain the value of a heap location independently of expressing ownership.
For example, $\code{x.f} = 5$ is a valid IDF assertion that expresses that $\code{x.f}$ stores the value $5$ without expressing ownership of \code{x.f}.
However, IDF requires assertions used as pre- and postconditions, loop invariants, frames (for the frame rule), etc.\ to be \emph{self-framing}, that is, to express ownership of all heap locations they mention. For example, $\acc{\code{x.f}} * \code{x.f} = 5$ is self-framing, while
$\code{x.f} = 5$ is not.
To capture IDF states with fractional permissions, we define the state model
$\Sigma_{\mathit{IDF}} \triangleq (L \rightharpoonup V) \times (L \rightharpoonup [0, 1])$.%
\footnote{In addition, values must exist for those heap locations where the state has non-zero permission. That is, $\Sigma_{\mathit{IDF}}$ is restricted to states $(h, \pi)$ such that $\forall l \ldotp \pi(l) > 0 \Rightarrow l \in \dom(h)$.}
In contrast to $\Sigma_{\mathit{SL}}$, values and permissions are separated in $\Sigma_{\mathit{IDF}}$,
which allows
states $(h, \pi)$ where $h(\code{x.f}) = 5$ but $\pi(\code{x.f}) = 0$.

We call a state $(h, \pi) \in \Sigma_{\mathit{IDF}}$ \emph{stable} iff it contains values exactly for the heap locations with non-zero permission, \ie $\dom(h) = \{ l \mid \pi(l) > 0 \}$.
Stable states are exactly those that can be represented as states in $\Sigma_{\mathit{SL}}$;
By construction, all states in $\Sigma_{\mathit{SL}}$ are stable.

\begin{figure}
    \small
\begin{mathpar}
    a \oplus b = b \oplus a \and
    a \oplus (b \oplus c) = (a \oplus b) \oplus c \and
    \thibault{c = a \oplus b} \land c = c \oplus c \Rightarrow a = a \oplus a \\
    x = x \oplus \core{x} \and
    \core{x} = \core{x} \oplus \core{x} \and
    x = x \oplus c \Rightarrow |x| \succeq c \and
    \core{a \oplus b} = \core{a} \oplus \core{b} \and
    \thibault{x \oplus a = x \oplus b} \land \core{a} = \core{b} \Rightarrow a = b \and
    \stable{\omega} \Rightarrow \omega = \stabilize{\omega} \and
    \stable{\stabilize{\omega}} \and
    \stabilize{a \oplus b} = \stabilize{a} \oplus \stabilize{b} \and
    x = \stabilize{x} \oplus \core{x} \and
    a = b \oplus \stabilize{\core{c}} \Rightarrow a = b
\end{mathpar}
\caption{Axioms for our \IDFAlgebra{} $(\Sigma, \oplus, \core{\_}, \stableNa, \stabilizeNa)$.
We define $(\omega' \succeq \omega) \triangleq (\exists r \ldotp \omega' = \omega \oplus r)$.}
\label{fig:IDF-algebra}
\end{figure}

To capture arbitrary SL and IDF states, we define an \emph{\IDFAlgebra{}} as follows:
\begin{definition}
    An \emph{\IDFAlgebra{}} is a tuple $(\Sigma, \oplus, \core{\_}, \stableNa, \stabilizeNa)$ that satisfies all axioms in \figref{fig:IDF-algebra},
    where $\Sigma$ is a set of states,
    $\oplus$ is a partial, commutative, and associative addition on $\Sigma$ (\ie a partial function from $\Sigma \times \Sigma$ to $\Sigma$),
    $\core{\_}$ 
    and $\stabilizeNa{}$ are endomorphisms of $\Sigma$,
    and $\stableNa$ is a predicate on $\Sigma$.
\end{definition}

The set $\Sigma$ and the partial addition $\oplus$ are the standard components of a separation algebra.
Using $\oplus$, we define the standard partial order $\succeq$ induced by $\oplus$ as $(\omega' \succeq \omega) \triangleq (\exists r \ldotp \omega' = \omega \oplus r)$.
We require \emph{positivity} ($c = a \oplus b = c \land c = c \oplus c \Rightarrow a = a \oplus a$) to ensure that the partial order is antisymmetric ($a \succeq b \land b \succeq a \Rightarrow a = b$).
Intuitively, the endomorphism $\core{\_}$ projects a state $\omega$ on its largest \emph{duplicable} part, \ie $|\omega|$ is the largest state smaller than $\omega$ such that $|\omega| = |\omega| \oplus |\omega|$.
Similarly, the endomorphism $\stabilizeNa{}$ projects a state $\omega$ on its largest \emph{stable} part, \ie $\stabilize{\omega}$ is the largest stable state smaller than $\omega$.

\paragraph{Instantiations}
For our concrete IDF state model $\Sigma_{\mathit{IDF}}$,
the combination $(h_1, \pi_1) \oplus (h_2, \pi_2)$ is defined iff $h_1$ and $h_2$ agree on the locations to which both states hold non-zero permission and the sums of their permissions pointwise is at most 1,
\ie
iff $\forall l \ldotp (\pi_1(l) + \pi_2(l) \leq 1) \land (l \in \dom(h_1) \cap \dom(h_2) \Rightarrow h_1(l) = h_2(l))$.
When the combination is defined, $(h_1, \pi_1) \oplus (h_2, \pi_2) \triangleq (h_1 \cup h_2, \pi_1 + \pi_2)$.
Knowledge about heap values is duplicable, whereas permissions are not.
Thus, $\core{\_}$ puts all permissions to $0$ but preserves the heap, \ie $\core{(h, \pi)} \triangleq (h, \lambda l \ldotp 0)$.
Moreover, $\stabilizeNa$ erases all values for heap locations to which the state does not hold any permission,
\ie $\stabilize{(h, \pi)} \triangleq ((\lambda l \ldotp \text{ if } \pi(l) > 0 \text{ then } h(l) \text{ else } \bot), \pi)$.

Separation algebra instances can also be instantiated as \IDFAlgebra{}s, by defining $\stableNa$ to be true for all states, and $\stabilizeNa$ to be the identity function on $\Sigma$.
For example, $\Sigma_{\mathit{SL}}$ (defined above) can be instantiated as an \IDFAlgebra{} with these definitions of $\stableNa$ and $\stabilizeNa$,
and with $\core{\_}$ mapping every state to the unit state (where all permissions are $0$, and the domain of the heap is empty).
Moreover, like separation algebras~\cite{Dockins2009, Iris1}, \IDFAlgebra{}s support standard constructions like the agreement algebra (where only $\omega = \omega \oplus \omega$ holds), and can be constructed by combining smaller algebras, via combinators such as product and sum types (where both types must be \IDFAlgebra{}s), function types (where only the codomain must be an \IDFAlgebra{}), etc.

\paragraph{State model for \theIVL{}}
Our \theIVL{} framework can be instantiated for any \IDFAlgebra{}.
We obtain the state model by extending this \IDFAlgebra{} with a store of local variables,
\ie a partial mapping from variables in $\mathit{Var}$ to values in $\mathit{Val}$.
Concretely,
given an \IDFAlgebra{} with carrier set $\Sigma$,
we define the state model for \theIVL{} as the product algebra $\Sigma_{\mathit{IVL}} \triangleq ((\mathit{Var} \rightharpoonup \mathit{Val}) \times \Sigma)$,
where the store $\mathit{Var} \rightharpoonup \mathit{Val}$ is instantiated to the agreement algebra, \ie addition on stores is defined for identical stores (as the identity).
Using the agreement algebra for the store ensures that $\vinhaleNa$ and $\vexhaleNa$ have no effect on local variables.

\paragraph{Self-framing IDF assertions}
Given an arbitrary \IDFAlgebra{}, we can define a general notion of \emph{self-framing} assertions and assertions \emph{framing} other assertions as follows.
\begin{definition}
In the following, $P$ is an IDF assertion (\ie a set of states from an \IDFAlgebra{}).
\begin{itemize}
\item  $P$ is \emph{self-framing}, written $\selfFraming{P}$, iff $\forall \omega \ldotp \omega \in P \Leftrightarrow \stabilize{\omega} \in P$.
\item A state $\omega$ \emph{frames} $P$, written $\relStable{\omega}{P}$, iff $\selfFraming{\{ \omega \} * P}$.
\item An IDF assertion $B$ \emph{frames} $P$, written $\framedBy{B}{P}$, iff $\forall \omega \in B \ldotp \stable{\omega} \Rightarrow \relStable{\omega}{P}$.
\item $P$ \emph{frames} an \emph{expression} (\ie a partial function from states to values) $e$, written $\framedBy{P}{e}$, iff $e(\omega) \text{ is defined}$ for all states $\omega \in P$.
\end{itemize}
\end{definition}

Those different notions are tightly connected:
If $A$ is self-framing and $A$ frames $B$ then $A * B$ is self-framing.
For example, the assertion $A \triangleq (\acc{\code{x.f}} * \code{x.f} = 5)$ is self-framing, because any state $\omega_A \in A$
has full permission to \code{x.f}, and thus $\stabilize{\omega_A}$ will retain the knowledge that \code{x.f} is $5$, and hence $\stabilize{\omega_A} \in A$.
In contrast, the assertion $B \triangleq (\code{x.f} = 5)$ is not self-framing,
since a state $\omega_B$ with no permission to \code{x.f} but with the knowledge that \code{x.f} is $5$ satisfies $B$,
but $\stabilize{\omega_B}$ will not retain the knowledge that $\code{x.f}=5$, and hence will not satisfy $B$.
Moreover, any state that satisfies $\acc{\code{x.f}}$ frames $B$,
thus the assertion  $\acc{\code{x.f}}$ frames $B$.
Note that, in an instantiation with SL states (\eg $\Sigma_{\mathit{SL}}$), all assertions are self-framing, since all SL states are stable.

\subsection{Operational Semantics}
\label{subsec:op-semantics}

We now formally define the operational semantics of \theIVL{} for the  state model described above
(given an arbitrary \IDFAlgebra{}).
As explained in \secref{subsec:op-semantics-key}, our operational semantics has judgments of the form
$\redStmt{C}{\omega}{S}$, where $\Delta$ is a type context,\footnote{%
In this section, we do not discuss typing in details, but our Isabelle formalization includes it.
In particular, it ensures that our operational and axiomatic semantics deal only with well-typed states, \ie states whose local store and heap contain values of the right types (defined by the type context $\Delta$).
By default, all states discussed in this section are well-typed.
}
$C$ is a statement, $\omega$ is a state, and $S$ is a set of states (to capture demonic non-determinism; angelic non-determinism is captured by the existence of different derivations $\redStmt{C}{\omega}{S_1}$ and $\redStmt{C}{\omega}{S_2}$).

\begin{figure}
\footnotesize
\begin{mathpar}
\inferhref{InhaleOp}{InhaleOpFull}
{\relStable{\omega}{A}}
{\redStmt{\vinhale{A}}{\omega}{ \{ \omega' \mid \exists \omega_A \in A \ldotp \omega' = \omega \oplus \omega_A \land \stable{\omega'} \} }}
%
\and
\inferhref{ExhaleOp}{ExhaleOpFull}
{\omega = \omega' \oplus \omega_A \\
\omega_A \in A \\
\stable{\omega'}}
{\redStmt{\vexhale{A}}{\omega}{ \{ \omega' \}}}
\and
\inferhref{SeqOp}{SeqOpFull}
{\redStmt{C_1}{\omega}{S_1} \\
\forall \omega_1 \in S_1 \ldotp \redStmt{\thibault{C_2}}{\thibault{\omega_1}}{\thibault{\mathcal{S}}(\omega_1)} }
{\redStmt{C_1;C_2}{\omega}{ \cup_{\omega_1 \in S_1} \thibault{\mathcal{S}}(\omega_1)}}
\and
\inferhref{AssignOp}{AssignOpFull}
{\Delta(x) = \tau \\
 e(\omega) = v \\
 v \in \tau}
{\redStmt{\vassign{x}{e}}{\omega}{\{ \omega[x \mapsto v] \}}}
\and
\inferhref{SkipOp}{SkipOpFull}
{}
{\redStmt{\vskipp}{\omega}{ \{ \omega \} }}
\and
\inferhref{HavocOp}{HavocOpFull}
{\Delta(x) = \tau}
{\redStmt{\vhavoc{x}}{\omega}{\{
\omega[x \mapsto v] \mid v \in \tau  \}}}
\quad
%
%
%
%
%
%
%
%
\inferhref{IfTOp}{IfTOpFull}
{b(\omega) = \top \\
\redStmt{C_1}{\omega}{S_1}}
{\redStmt{\vif{b}{C_1}{C_2}}{\omega}{S_1}}
\quad
%
\inferhref{IfFOp}{IfFOpFull}
{b(\omega) = \bot \\
\redStmt{C_2}{\omega}{S_2}}
{\redStmt{\vif{b}{C_1}{C_2}}{\omega}{S_2}}
\end{mathpar}
\caption{Operational semantics rules.
}
\label{fig:op-semantics}
\end{figure}

The rules for the operational semantics are given in \figref{fig:op-semantics}.
%
As shown by the rule \opNameFull{Inhale},
$\vinhale{A}$ can reduce in a state $\omega$ only if $\omega$ frames $A$.
In our concrete instantiation $\Sigma_{\mathit{IDF}}$, this means that $\omega$ or $A$ must contain the permission to any heap location mentioned in $A$.
For example, $\vinhale{\code{x.f} = 5}$ can reduce correctly only in a state $\omega$ that has some permission to \code{x.f}.
If $\omega$ has a different value than $5$ for \code{x.f}, the statement will reduce to an empty set of states,
\ie $\redStmt{\vinhale{\code{x.f} = 5}}{\omega}{\varnothing}$, capturing the fact that we inhaled an assumption inconsistent with our state.
In this case, the rest of the program is trivially correct (because it will be executed in no state).
If $\omega$ has value $5$ for \code{x.f}, then the statement will reduce to the singleton set $\{ \omega \}$,
\ie $\redStmt{\vinhale{\code{x.f} = 5}}{\omega}{\{ \omega \}}$.
Finally, inhaling $\acc{\code{x.f}}$ in a state $\omega$ with no permission and no value to \vcode{x.f} will result in a set with multiple states (potentially infinitely many), one state for each possible value of \vcode{x.f}.
We require $\stable{\omega'}$ in the rule to ensure that executing a statement in any stable state leads to a set of stable states,
\ie $\forall \omega \ldotp \stable{\omega} \land \redStmt{\omega}{C}{S} \Rightarrow (\forall \omega' \in S \ldotp \stable{\omega'})$.
In other words, the operational semantics preserves the stability of states.

Dually, the rule \opNameFull{Exhale} requires the final state $\omega'$ to be stable.
This ensures that values of heap locations for which the state lost all permission will be erased.
For example, $\vexhale{\acc{x.f}}$ succeeds only in a state with full permission to \vcode{x.f}, and results in a final state
without any permission nor value for \vcode{x.f}.
Note that the rule \opNameFull{Exhale} is the only atomic rule that uses angelic nondeterminism, because the rule can be applied with different $\omega'$ (corresponding to different angelic choices).
(The rules \opNameFull{Inhale} and \opNameFull{Havoc} use demonic non-determinism, while \opNameFull{Assign} and \opNameFull{Skip} are deterministic.)
The rule \opNameFull{Seq} first executes $C_1$ in $\omega$, which yields a set of states $S_1$.
Since $S_1$ captures demonic choices, $C_2$ must be executed in \emph{all} states from $S_1$, but the angelism in $C_2$ can be resolved differently for each state, which is captured by the choice of the function $\mathcal{S}$.
The function $\mathcal{S}$ must map every state $\omega_1$ from $S_1$ (the set of states obtained after executing $C_1$ in $\omega$) to a set of states $\mathcal{S}(\omega_1)$ that can be reached by executing $C_2$ in $\omega_1$.

Finally, note that expressions in \theIVL{} are \emph{semantic}, \ie they are \emph{partial} functions from states to values.
We model them as partial functions because they might be heap-dependent, and thus might not be defined for all states.
For example, the expression $\code{x.f} = 5$ is only meaningful in states where \code{x.f} has a value.
The rules \opNameFull{Assign}, \opNameFull{IfT}, and \opNameFull{IfF} require that the expressions are defined in the initial state $\omega$.

\subsection{Axiomatic Semantics}
\label{subsec:ax-semantics}

\begin{figure}
\footnotesize
\begin{mathpar}
\inferhref{SkipAx}{SkipAxFull}
{\selfFraming{P}}
{\axiomSem{P}{\vskipp}{P}}
\and
\inferhref{InhaleAx}{InhaleAxFull}
{\selfFraming{P} \\
 \framedBy{A}{P}}
{\axiomSem{P}{\vinhale{A}}{P * A }}
\and
\inferhref{ExhaleAx}{ExhaleAxFull}
{\selfFraming{P} \\
 P \models Q * A \\
 \selfFraming{Q}}
{\axiomSem{P}{\vexhale{A}}{Q}}
\and
\inferhref{IfAx}{IfAxFull}
{\selfFraming{P} \\
 \framedBy{P}{b} \\
 \axiomSem{P \wedge b}{C_1}{B_1} \\
 \axiomSem{P \wedge \lnot b}{C_2}{B_2}}
{\axiomSem{P}{\vif{b}{C_1}{C_2}}{B_1 \lor B_2}}
\quad
%
%
%
%
%
\inferhref{HavocAx}{HavocAxFull}
{\selfFraming{P} \\
 \Delta(x) = \tau}
{\axiomSem{P}{\vhavoc{x}}{\exists x \in \tau \ldotp P}}
\and
\inferhref{SeqAx}{SeqAxFull}
{\axiomSem{P}{C_1}{R} \\
 \axiomSem{R}{C_2}{Q}}
{\axiomSem{P}{C_1;C_2}{Q}}
\and
\inferhref{AssignAx}{AssignAxFull}
{\selfFraming{P} \\
 \framedBy{P}{e}}
{\axiomSem{P}{\vassign{x}{e}}{ \exists v \ldotp
\substitution{P}{v}{x} \land
x = \substitution{e}{v}{x} }}
%
\end{mathpar}
\caption{Axiomatic semantic rules.
}
\label{fig:ax-semantics}
\end{figure}

Using the same extended state model as in the operational semantics, we define an axiomatic semantics with judgments of the form
$\axiomSem{P}{C}{Q}$, where $\Delta$ is a type context, $P$ and $Q$ are assertions (sets of states), and $C$ is a \theIVL{} statement.
All rules are shown in \figref{fig:ax-semantics}.
Multiple rules have side-conditions requiring the preconditions and postconditions to be self-framing,
ensuring that if we have $\axiomSem{P}{C}{Q}$, $P$ and $Q$ are self-framing.

As explained in \secref{subsec:ax-semantics-key}, our operational and axiomatic semantics are equivalent. The soundness property expressed in \thmref{thm:op-to-ax-simple} (in \secref{subsec:ax-semantics-key}) allows one to bridge the gap between a valid \theIVL{} program (according to \defref{def:correct}) and the front-end program logic.
The proof of \thmref{thm:op-to-ax-simple} is not straightforward.
In particular, our proof explicitly tracks the angelic choices made based on the sequence of past states of each execution,
as shown by the following lemma, which implies \thmref{thm:op-to-ax-simple}.
Let $\Stabilize{A} \triangleq \{ \omega' \mid \stabilize{\omega'} \in A \}$ for a set of states $A$.
\newcommand{\Sfam}{\mathcal{S}}
\newcommand{\Shist}{\Omega}
\begin{lemma}
  \label{lem:op-to-ax-ind}
  Given a set $\Shist \in \powerset{ \Sigma^* \times \Sigma}$ of lists of past states paired with current states, a \theIVL{} statement $C$,
  and a function $\Sfam$ mapping
  elements from $\Omega$ to sets of states,
  if for all $(l, \omega) \in \Shist$ we have
  \begin{enumerate}
    \item $\stable{\omega}$, and
    \item $\redStmt{C}{\omega}{ \Sfam(l, \omega) }$,
  \end{enumerate}
    then $\axiomSem{\Stabilize{ \{ \omega \mid (l, \omega) \in \Shist \} }}{C}{\Stabilize{ \bigcup_{ (l, \omega) \in \Shist} \Sfam(l, \omega) }}$.
\end{lemma}


An element $([\omega_0, \ldots, \omega_n], \omega_{n+1}) \in \Shist$ represents all the intermediate states of one execution up to now,
which we use to resolve the future angelism.
The function $\Sfam$ maps each such element to a set of states that can be reached from $\omega_{n+1}$ by executing $C$.
Intuitively, the precondition collects all the current states from $\Omega$, and the postcondition collects all the states they can reach by executing $C$.
The proof proceeds by structural induction over the statement $C$.

\paragraph{The reason for tracking sequences of past states}
The reader might be wondering why \lemref{lem:op-to-ax-ind} keeps track of the list of all past states, instead of only keeping track of the current state.
The reason is that only keeping track of the current state would not
allow us to prove the \thibault{inductive} case for sequential composition.
To understand why, \emph{assume} that
$\Shist$ is \emph{instead} a set of \emph{single states},
and $\Sfam$ is a function from \emph{single states} to a set of states.
\thibault{Consider proving the inductive case for the sequential composition, \ie for $C \triangleq (C_1; C_2)$.}
\thibault{Assume that, in this scenario, we are given $\Shist = \{ \omega_A, \omega_B \}$, and that $\Sfam$ is such that $\Sfam(\omega_A) = \{ \omega'_A \}$
and $\Sfam(\omega_B) = \{ \omega'_B \}$.}
\thibault{From assumption (2) in \lemref{lem:op-to-ax-ind}, we know that}
$\redStmt{C_1; C_2}{\omega_A}{ \{ \omega'_A \} }$ and $\redStmt{C_1; C_2}{\omega_B}{ \{ \omega'_B \} }$ hold.
It might be the case that executing $C_1$ in either $\omega_A$ or $\omega_B$ yields the same set of states $\{ \omega' \}$,
\ie $\redStmt{C_1}{\omega_A}{ \{ \omega' \} }$ and $\redStmt{C_1}{\omega_B}{ \{ \omega' \} }$,
but that the angelic non-determinism when executing $C_2$ in state $\omega'$ was resolved differently in both executions, leading to $\omega'_A$ in the execution from $\omega_A$ and $\omega'_B$ in the execution from $\omega_B$.
More concisely, the executions of $C_1; C_2$ in $\omega_A$ and $\omega_B$ might have been constructed as follows:
\begin{align*}
    \redStmt{C_1}{\omega_A}{ \{ \omega' \} } \land \redStmt{C_2}{\omega'}{ \{ \omega'_A \} } \Rightarrow \redStmt{C_1; C_2}{\omega_A}{ \{ \omega'_A \} } \\
    \redStmt{C_1}{\omega_B}{ \{ \omega' \} } \land \redStmt{C_2}{\omega'}{ \{ \omega'_B \} } \Rightarrow \redStmt{C_1; C_2}{\omega_B}{ \{ \omega'_B \} }
\end{align*}
In this case, our intermediate set of states between $C_1$ and $C_2$ is $\Shist \triangleq \{ \omega' \}$.
To apply our induction hypothesis for $C_2$,
we need to find a function $\Sfam_2$ that maps $\omega'$ to both $\{ \omega'_A \}$ and $\{ \omega'_B \}$, as required by the assumption (2) in the lemma, which is not possible.

To solve this issue, we explicitly keep track of all past states.
In this way, our intermediate set of states for the previous example is $\Shist \triangleq \{ ([\omega_A], \omega'), ([\omega_B], \omega') \}$,
which allows us to define a function $\Sfam_2$ such that $\Sfam_2([\omega_A], \omega') = \{ \omega'_A \}$ and $\Sfam_2([\omega_B], \omega') = \{ \omega'_B \}$,
allowing us to apply our induction hypothesis and prove the sequential composition case.

\paragraph{Completeness}
To show that our operational and axiomatic semantics are equivalent, we also prove the following completeness property
(whose proof is less involved than for soundness):
\begin{theorem}[Completeness]\label{thm:completeness}
    Assume $\axiomSem{P}{C}{Q}$,
    and let $\omega \in P$ such that $\stable{\omega}$.
    Then there exists $S$ such that $\redStmt{C}{\omega}{S}$ and $S \subseteq Q$.
\end{theorem}


\subsection{\ViperInst{}: Instantiating \theIVL{} with Viper}
\label{sec:viper-inst}

To show the practical usefulness of \theIVL{}, we instantiated it for the Viper language.
We call this instantiation \ViperInst{}, and we use it in \secref{sec:backends} and \secref{sec:frontends}.
To instantiate the framework presented in this section, one needs
(1) an \IDFAlgebra{},
(2) a type of custom statements $C'$,
(3) operational and axiomatic semantic rules for each custom statement,
and
(4) proofs that those operational and axiomatic semantic rules are compatible with our framework
(\ie soundness and completeness for the custom semantic rules, and a proof that the operational semantics of custom statements preserves stability).

We instantiate (1) with the \IDFAlgebra{} $\Sigma_{\mathit{IDF}}$ defined in \secref{subsec:idf-algebra}, where
the set $L$ of heap locations is the set of pairs of a reference and a field (represented by a string).
For (2), we add field assignments as $C' \Coloneqq (e_1.f := e_2)$, where $e_1$ and $e_2$ are semantic expressions that evaluate to a reference and a value, respectively, and $f$ is a field.
The field assignment $e_1.f := e_2$ is deterministic.
In an initial state $(\sigma, (h, \pi))$, it reduces to the singleton set $\{ (\sigma, (h[(r, f) \mapsto v], \pi)) \} $
if $e_1$ evaluates to a reference $r$, $e_2$ evaluates to a value $v$, and $\pi((r, f)) = 1$.
This semantics is reflected both in its corresponding operational and axiomatic semantic rules (3),
and the associated proofs (4) are straightforward.

Moreover, we have also connected the deep embedding of the Viper language developed by \citet{THSem}
(which we leverage in the next section) to \ViperInst{},
by defining a function $\IVLcompile{\stmt}$ that converts their \emph{syntactic} statements, expressions and assertions into \emph{semantic} \ViperInst{} statements, expressions and assertions.
\section{Back-End Soundness}
\label{sec:backends}

In this section, we show how our framework enables formalizing the soundness of different back-end verifiers.
We prove the soundness of two fundamentally different verification algorithms commonly used in practice: symbolic execution and verification condition generation.
We connect both to the \emph{same} instantiation of \theIVL{}, namely \ViperInst{} introduced in~\secref{sec:viper-inst}. This demonstrates that \theIVL{}'s semantics can accommodate fundamentally different verification algorithms.

Symbolic execution is a common kind of verification algorithm used in separation logic-based verifiers~\cite{BerdineCO05,VeriFast, GillianI}.
\secref{sec:symbolic-execution} introduces a symbolic execution back-end for \ViperInst{}.
Its design follows Viper's symbolic execution back-end~\cite{MalteThesis}, but it is formalized as a function inside \isabelle{}\@.
The main result of \secref{sec:symbolic-execution} is a soundness proof of this symbolic execution against the operational semantics of \ViperInst{}, showing how \theIVL{} is general enough to justify widely-used symbolic execution algorithms.

In \secref{sec:verif-cond-gener} we connect \ViperInst{} to the formalization by
\citet{THSem} of Viper's verification condition generation (VCG) back-end, which translates an input Viper program to Boogie.\footnote{This work provides a proof-producing version of Viper's VCG back-end that generates a certificate in Isabelle for each successful verification, but not a general soundness proof.} This formalization includes a formal operational semantics of Viper that we call \emph{\THSem{}}.
Unlike \ViperInst{}, which is designed to capture the verification algorithms of multiple back-ends, \THSem{} is specific to the verification algorithm of the VCG back-end.
For example, \THSem{} uses a total heap (\ie all possible locations on the heap store a value), while \ViperInst{} is based on a partial heap (which is important to capture existing symbolic execution algorithms). Moreover,
\THSem{} uses (constrained) \emph{demonic} choice when exhaling wildcard permissions, while \ViperInst{} uses angelic choice.
Despite these differences, we show that \ViperInst{}'s (and thus also \theIVL{}'s) operational semantics is general enough to capture \THSem{}, which embodies Viper's VCG back-end.

We have chosen these two back-ends since they implement very different proof search algorithms:
The symbolic execution algorithm manipulates a \emph{symbolic state} including a list of heap chunks, while the VCG back-end maps to Boogie code whose operations are embodied by \THSem{}, a big-step operational semantics with a total heap.
These back-ends show \theIVL{}'s generality for justifying multiple common verification algorithms.
A key aspect that enables this generality is \theIVL{}'s use of angelic choice.
Concretely, the two back-ends use different algorithms for exhaling wildcard permissions (the symbolic execution halves the permission of one heap chunk while \THSem{} \emph{demonically} chooses a suitably-constrained permission amount).
Yet, \theIVL{} can capture both algorithms thanks to its use of angelic choice.

\subsection{Symbolic Execution}
\label{sec:symbolic-execution}

\begin{figure}
	\small
  \centering
  \begin{align*}
  \SState : \SymState &\bnfdef \{\symstore : \Variable \rightharpoonup \SymExpr, \sympc : \SymExpr, \symheap : \ListTyp{\Chunk}\} \\
  \symexpr &\bnfdef \symvar \mid \literal \mid \unoperation~\symexpr \mid \symexpr \binoperation \symexpr \text{ where } \unoperation \in \{\neg, -, \ldots\} \text{ and } \binoperation \in \{\wedge, \vee, =, +, -, \ldots\} \\
  \chunk : \Chunk &\bnfdef \{\chunkrecv : \SymExpr, \chunkfield : \FieldName, \chunkperm : \SymExpr, \chunkval : \SymExpr\}
\end{align*}
  \begin{align*}
    \sexec~\SState~\stmt~\SCond &\eqdef \left\{
    \begin{aligned}
      &\sproduce~\SState~\assertion~\SCond & \text{ if } \stmt = \Inhale{\assertion}\\
      &\sconsume~\SState~\assertion~(\Lam \SState. \sstablize~\SState~\SCond) & \text{ if } \stmt = \Exhale{\assertion}\\
      &\sexp~\SState~\expr~(\Lam \SState~\symexpr. \sexec~\symaddcond{\SState}{\symexpr}~\stmt_1~\SCond & \text{ if } \stmt = (\SynIf{\expr}{\stmt_1}{\stmt_2})
      \\ &\quad \wedge \sexec~\symaddcond{\SState}{\neg\symexpr}~\stmt_2~\SCond)
     \\
      & \ldots &
    \end{aligned}\right.
  \end{align*}
  \[
    \sproduce~\SState~(\accmath{\expr_r}{\fieldname}{\expr_p})~\SCond \eqdef\sexp~\SState~\expr_r~(\Lam\SState~\symexpr_r.
        \sexp~\SState~\expr_p~(\Lam\SState~\symexpr_p.
      \symaddchunk{\SState}{\{\symexpr_r, \fieldname, \symexpr_p, \symfresh\}}~\SCond))
  \]
  \[
    \sconsume~\SState~(\accmath{\expr_r}{\fieldname}{\wildcardmath})~\SCond \eqdef
      \sexp~\SState~\expr_r~(\Lam\SState~\symexpr_r.
        \symextract{\SState}{\symexpr_r}{\fieldname}{\wildcardmath}~
            (\Lam\SState~\chunk. \symaddchunk{\SState}{\chunk\{\chunkperm :=\chunk.\chunkperm / 2\}}~\SCond))
          \]
  \caption{Symbolic states and excerpts of $\sexec$, $\sproduce$, and $\sconsume$.
The full definition is in \appendixref{sec:sexec-full}{A}.}
  \label{fig:sexec}
\end{figure}

We formalized a symbolic execution back-end for \ViperInst{} in \isabelle{} based on the description of Viper's back-end by \citet{MalteThesis} while also taking inspiration from the (on paper) formalization of symbolic execution by \citet{GradualViper}.

\paragraph{Symbolic states}
The symbolic state tracked during verification is defined in \figref{fig:sexec}. It consists of the following components:%
\footnote{For simplicity, we omit components for generating fresh symbolic variables and tracking type information.}
(1) A \emph{symbolic store} ($\symstore$) mapping variables to symbolic expressions,
(2) a \emph{path condition} ($\sympc$)---which is a symbolic expression tracking logical facts that hold in the current branch of the program---, and
(3) a \emph{symbolic heap} ($\symheap$) given by a list of heap chunks.
\emph{Symbolic expressions $\symexpr$} consist of symbolic variables $\symvar$, literals $\literal$ (\eg for concrete integers, booleans or permission amounts), unary operations $\unoperation~\symexpr$, and binary operations $\symexpr \binoperation \symexpr$.
We define a function $\symaddcond{\SState}{\symexpr}$ that adds the (boolean) symbolic expression $\symexpr$ to the path condition of $\SState$.

The most crucial part of symbolic states is the symbolic heap. As is common~\cite{Smallfoot, VeriFast, MalteThesis}, we represent the symbolic heap as a list of (heap) chunks.
Conceptually, each heap chunk corresponds to an $\accmath{\expr_r}{\fieldname}{\expr_p}$ resource, which we call an $\mathsf{acc}$-resource in this section, together with an associated value.
Concretely, a chunk $\chunk$ is a record with four fields, as shown in \figref{fig:sexec}.
$\chunkrecv$ and $\chunkfield$ describe the heap location that the chunk belongs to, $\chunkperm$ describes the permission of the chunk, and $\chunkval$ gives the (symbolic) value of the heap location.
A symbolic heap is a list of chunks. Note that this list can contain multiple chunks for the same location (\cf{} state consolidation, described shortly).

\paragraph{Defining the symbolic execution}
Our symbolic execution is defined via the $\sexec$ function for symbolically executing a statement $\stmt$. It delegates calls to: the $\sproduce$ function (for inhaling an assertion $\assertion$), the $\sconsume$ function (for exhaling an assertion $\assertion$), the $\sstablize$ function (for removing empty heap chunks after exhaling an assertion), and the $\sexp$ function (for symbolically evaluating an expression $\expr$). Each of these functions are formalized as functions in \isabelle{} and can be executed inside the prover to verify a concrete program.
The parts of these functions relevant to this chapter are shown in \figref{fig:sexec}. The full definition can be found in \appendixref{sec:sexec-full}{A}.
Following \citet{MalteThesis}, these functions are written in continuation passing style with continuation $\SCond$. This allows us to easily split the verification in multiple branches as shown \eg by the \textsf{if}-case of $\sexec$.
We now highlight the most important aspects of the symbolic execution.

\paragraph{Representing different state consolidation algorithms}
After inhaling an $\mathsf{acc}$-resource and adding it to the list of heap chunks, the symbolic execution might try to merge chunks for the same location and deduce additional information (\eg that for chunks of the same location their permissions sum to at most $1$ and their values match).
This process, called \emph{state consolidation}~\cite{MalteThesis}, is
incorporated into the $\symaddchunkH$ function, used to model inhaling an $\mathsf{acc}$-resource during $\sproduce$:
\[
  \symaddchunk{\SState}{\chunk}~\SCond\eqdef \symconsolidate{\SState\{\symheap:=\chunk::\SState.\symheap\}}~\SCond
\]
Since there are many possible ways to implement state consolidation~\cite{MalteThesis} (\eg merging chunks eagerly or lazily), we do not prescribe a specific implementation of the $\symconsolidateH$ function, but instead characterize
$\symconsolidateH$ \emph{semantically}:%
\footnote{The actual definition of $\symconsolidateH$ is slightly different to decouple the definition of the symbolic execution and \ViperInst{}.}
\[
  \symconsolidate{\SState}~\SCond\eqdef
  \All \IVLstate. \symrel{\IVLstate}{\SState} \Rightarrow \Exists \SState'.
  \symrel{\IVLstate}{\SState'} \wedge \SCond~\SState'
\]
Concretely, when executing $\symconsolidateH$, one is given a \ViperInst{} state $\IVLstate$ related to the current symbolic execution state $\SState$ (using the $\symrel{}{}$ relation) and one can pick an arbitrary new state $\SState'$ as long as it is related to the same \ViperInst{} state $\IVLstate$.
Intuitively, $\symrel{\IVLstate}{\SState}$ is defined by stating that there exists a mapping from symbolic variables to concrete values, which can be simply extended to a mapping from $\SState$ to $\IVLstate$.
The existential quantifier 
allows us to represent many different state consolidation algorithms. However, this generality also means that $\symconsolidateH$ cannot be executed directly. Instead, one can provide a concrete algorithm and prove it sound against $\symconsolidateH$ (our implementation uses the trivial algorithm that does not consolidate at all).
However, the soundness proof of our symbolic execution works for any valid consolidation algorithm.

\paragraph{Soundness}

We prove $\sexec$ sound against the operational semantics of \ViperInst{}:%
\footnote{We omit side-conditions about typing to avoid clutter.}
\begin{theorem}[Soundness of $\sexec$]
  \label{thm:sexec-sound}
  For each (syntactic) statement $\stmt$, \ViperInst{} state $\IVLstate$ and symbolic state $\SState$ related via $\symrel{\IVLstate}{\SState}$, if $\sexec~\SState~\stmt~\SCond$ evaluates to true, then $\IVLcompile{\stmt}$ is correct for the initial state $\IVLstate$.
\end{theorem}
$\IVLcompile{\stmt}$ is the compilation function from syntactic statements to \ViperInst{} statements described in \secref{sec:viper-inst}.
The operational semantics of \theIVL{} is well-suited for this soundness proof since the symbolic execution also traverses the statements in an operational way, and it is straightforward to relate one \ViperInst{} state to one symbolic execution state via $\symrel{\IVLstate}{\SState}$.


\paragraph{Soundness of exhaling wildcards via angelic choice}
Let us highlight the most interesting part of this soundness proof: exhaling wildcards.
Exhaling assertions is handled by the $\sconsume$ function in \figref{fig:sexec}.
When exhaling an $\mathsf{acc}$-resource with a wildcard permission amount, $\sconsume$ finds and removes a matching chunk from the symbolic heap using the $\symextractH$ function.%
\footnote{Similar to $\symconsolidateH$, $\symextractH$ is characterized semantically and we provide a default implementation of $\symextractH$ that queries \isabelle{}'s solvers to find the first matching chunk.}
Then it adds the chunk back with its permission amount halved.
Representing this algorithm directly in \ViperInst{} would be impossible since there might be multiple heap chunks for the same location and thus
the amount of permissions removed depends on the structure of the symbolic heap.
This structure is not visible in \ViperInst{}, which tracks only a single concrete heap.
However, we can still prove this algorithm sound against \ViperInst{}. The angelic choice in the operational semantics allows us to pick \emph{any} non-zero permission amount to remove when constructing the \ViperInst{} execution, in particular, the amount that was chosen by the execution of $\sexec$.
This shows how angelic choice gives \theIVL{} the flexibility to be used in the soundness proof for different verification algorithms, even some that cannot be represented directly in the \theIVL{}.


\subsection{Verification Condition Generation}
\label{sec:verif-cond-gener}

We now describe how we connect the \ViperInst{} instantiation of \theIVL{} to the \THSem{}  formalization of Viper's VCG introduced by \citet{THSem}. \THSem{} is expressed as an operational big-step semantics $\THSemRed{\THSemCtx}{\THSemStmt}{\THSemState}{\THSemRes}$.
Here, $\THSemStmt$ is a (deeply embedded) Viper statement, $\THSemState$ the initial \THSem{} state consisting of a total heap (mapping all locations to values) and a permission mask (mapping all locations to permission amounts), and $\THSemRes$ is an \emph{outcome}, which can be either \emph{failure} $\THSemFail$, \emph{magic} $\THSemMagic$, or a \emph{normal outcome} $\THSemNormal{\THSemStateB}$.%
\footnote{We omit typing contexts in this section to avoid clutter.}
The key result of \citet{THSem} is that for each successful verification run of the VCG algorithm, they provide a proof that the \THSem{} execution does not fail:
$\neg (\THSemRed{\THSemCtx}{\THSemStmt}{\THSemState}{\THSemFail})$.

What makes the connection between \THSem{} and \ViperInst{} interesting is that
\THSem{} makes various design choices that are specific to the Viper back-end that it was designed to represent.
For instance, \THSem{} defines the \vexhaleNa{} of a wildcard to demonically remove a non-zero permission amount smaller than the currently held amount, which precisely mimics Viper's VCG.
Morever, \THSem{} chooses a total heap representation for the Viper states, where \emph{all} locations store a value (\THSem{} checks that only locations with non-zero permission are accessed), because this is how Viper's VCG back-end represents the heap.
In contrast, \ViperInst{} uses a more standard partial heap introduced in \secref{subsec:idf-algebra}.
By proving \THSem{} sound against \ViperInst{}, we show that \theIVL{} as a general semantics for verification algorithms can capture this preexisting verification algorithm.
The most significant challenge in the proof connecting \THSem{} and \ViperInst{} is the difference in their heap representations. We explain this challenge and our solutions next.


\paragraph{Total vs.\ partial heap}
The seemingly superficial difference between \THSem{}'s total heap and \ViperInst{}'s partial heap has far-reaching ramifications:
fundamentally it means that a \ViperInst{} execution does not correspond to a single \THSem{} execution but a \emph{set} of \THSem{} executions.

The reason for this mismatch is in the semantics of $\vexhaleNa{}$.
When exhaling all permissions to a location and later inhaling permissions to this location again, a Viper semantics needs to pick a \emph{fresh} value for the location such that one cannot unsoundly assume that the value remained unchanged between the $\vinhaleNa{}$ and the $\vexhaleNa{}$.
This requirement is naturally expressed with the partial heap of \ViperInst{}: when exhaling all permissions to a location in \ViperInst{}, the location is removed from the partial heap and when new permissions for the location are inhaled, it gets re-added with  a (non-deterministically chosen) fresh value.
However, since \THSem{} uses a total heap, it cannot \emph{remove} locations.
Instead, \THSem{} non-deterministically assigns these locations new values \emph{after the $\vexhaleNa{}$} and leaves the heap unchanged in the $\vinhaleNa{}$.
Consequently, \THSem{} and \ViperInst{} apply (demonic) non-deterministic choice at different program points:
\THSem{} already picks a fresh value during the $\vexhaleNa{}$, while \ViperInst{} chooses it during the $\vinhaleNa{}$.
To address this mismatch\footnote{The mismatch could also be addressed by changing \THSem{} to assign a fresh value during $\vinhaleNa{}$. However, our goal is to capture the verification algorithms of \emph{existing} back-ends.}, we relate a \ViperInst{} execution not to a single \THSem{} execution but to a set of \THSem{} executions that represent all possible choices for the demonic non-determinism.

\paragraph{Soundness}
We prove the following soundness statement for \THSem{}:%
\footnote{For readability, we omit some technical assumptions about stability of $\IVLstate$ and well-typedness.}
\begin{theorem}[Soundness of \THSem{}]
  \label{thm:thsem-sound}
  For all (syntactic) statements $\THSemStmt$ and \ViperInst{} states $\IVLstate$, if we have $\neg (\THSemRed{\THSemCtx}{\THSemStmt}{\THSemState}{\THSemFail})$ for all \THSem{} states $\THSemState$ such that $\vcgrel{\IVLstate}{\THSemState}$, then $\IVLcompile{\THSemStmt}$ is correct for the state $\IVLstate$.
\end{theorem}

Intuitively, this theorem allows us to transform a proof about a successful verification by the VCG back-end into a verification proof according to the \ViperInst{} semantics.
\michael{Note that the theorem relates \peter{a} single \ViperInst{} execution 
to a \emph{set} of \THSem{} executions 
since the relation $\vcgrel{\IVLstate}{\THSemState}$ relates a \ViperInst{} state $\IVLstate$ to multiple \THSem{} states $\THSemState$ representing the different choices for the \peter{demonic} non-determinism.
(Otherwise, $\vcgrel{\IVLstate}{\THSemState}$ is similar to $\symrel{\IVLstate}{\SState}$ from \secref{sec:symbolic-execution}, but adapted for the different notion of states used by \THSem{}.)
}
\mout{
Note that the theorem relates the \ViperInst{} state to a \emph{set} of \THSem{} states.}
In fact, to prove \thmref{thm:thsem-sound} via induction, we need to prove a stronger lemma that also requires us to construct all possible \THSem{} executions for the statement corresponding to the \ViperInst{} execution.

\paragraph{Summary}
We have demonstrated in this section how \theIVL{}'s operational semantics helps us solve Challenge 2,
by being
general enough to capture the two predominant verification algorithms back-ends implemented in practice: our new formalization of symbolic execution in \secref{sec:symbolic-execution} and the preexisting formalization of Viper's VCG back-end~\cite{THSem} in \secref{sec:verif-cond-gener}.

\section{Front-End Soundness}
\label{sec:frontends}

\newcommand{\parImp}{ParImp}

In this section, we show how our axiomatic semantics addresses Challenge 3 from \secref{sec:introduction},
by formalizing and proving sound a concrete front-end translation into \ViperInst{} for a parallel programming language \parImp{} with
loops, shared memory, and dynamic memory allocation and deallocation.
We define the language and an IDF-based program logic in  \secref{subsec:csl}.
In \secref{subsec:translation}, we define the translation of annotated \parImp{} programs into \ViperInst{} and prove it sound using the axiomatic semantics of \ViperInst{}.
While the soundness proof is specific to this translation, it highlights key reusable ingredients and demonstrates how our axiomatic semantics for \theIVL{} makes such proofs simple.

\subsection{An IDF-Based Concurrent Separation Logic}
\label{subsec:csl}

Our parallel programming language \parImp{} is defined as%
\def\manskip{\kern-0.8pt}
{%
\footnotesize
$$
C \Coloneqq
x\manskip:=\manskip{}e \mid
x\manskip:=\manskip{}r.v \mid
r.v\manskip:=\manskip{}e \mid
r\manskip:=\manskip\fenew{e} \mid
\fefree{r} \mid
C;\manskip{}C \mid
\feif{b}{C}{C} \mid
C\manskip\parcomp\manskip{}C \mid
\fewhile{b}{C} \mid
\feskip
$$
}
$C$ ranges over \parImp{} statements, $e$ over arithmetic expressions, $b$ over boolean expressions, $x$ over integer variables, $r$ over reference variables, and $v$ is a fixed field.
We consider objects with a unique field $v$ for simplicity; extending our work to support multiple fields is straightforward.
The statement $x := r.v$ loads the value of the field $v$ of the reference $r$ into the variable $x$,
while $r.v := e$ stores the value of the expression $e$ in the field $v$ of the reference $r$.
The statement $r := \fenew{e}$ allocates a new reference with the value of the expression $e$ for the field $v$,
and $\fefree{r}$ deallocates the reference $r$.
The other statements are standard.
We use a standard small-step semantics
$\langle C, \sigma \rangle \rightarrow \langle C', \sigma' \rangle$
where $\sigma$ and $\sigma'$ are pairs of a store (a partial mapping from variables to values) and a heap (a partial mapping from pairs of an address and a field to values).

\begin{figure}
\footnotesize
\begin{mathpar}
    \inferhref{Frame}{Frame2}
    {\CSL{P}{C}{Q} \\ \selfFraming{P} \\ \selfFraming{F} \\
    \freevars{F} \cap \modvars{C} = \emptyset}
    {\CSL{P * F}{C}{Q * F}}
    \and
    \inferhref{Par}{Par3}
    {
    \modvars{C_l} \cap (\freevars{C_r} \cup \freevars{Q_r}) = \varnothing \\
    \modvars{C_r} \cap (\freevars{C_l} \cup \freevars{Q_l}) = \varnothing \\
    \CSL{P_l}{C_l}{Q_l} \\ \CSL{P_r}{C_r}{Q_r} \\
    \selfFraming{P_l} \\ \selfFraming{P_r}
    }
    {\CSL{P_l * P_r}{C_l \; || \; C_r}{Q_l * Q_r}}
    \and
    \inferhref{Seq}{Seq2}
    {\CSL{P}{C_1}{R} \\ \CSL{R}{C_2}{Q}}
    {\CSL{P}{C_1;C_2}{Q}}
    \and
    \inferhref{Cons}{Cons2}
    {\CSL{P'}{C}{Q'} \\ P \models P' \\ Q' \models Q}
    {\CSL{P}{C}{Q}}
    \and
    \inferhref{If}{If}
    {\CSL{P \land b}{C_1}{Q} \\ \CSL{P \land \neg b}{C_2}{Q}}
    {\CSL{P}{\feif{b}{C_1}{C_2}}{Q}}
    \and
    \inferhref{Alloc}{Alloc}
    {r \notin \freevars{e}}
    {\CSL{\top}{\code{r} := \fenew{\code{e}}}{ \acc{\code{r.v}} * \code{r.v} = \code{e}}}
    \and
    \inferhref{While}{While}
    {\CSL{I \land b}{C}{I}}
    {\CSL{I}{\fewhile{b}{C}}{I \land \neg b}}
    \and
    \inferH{Load}
    {P \models \wacc{\code{r.v}}}
    {\CSL{P}{\code{x := r.v}}{\exists u \ldotp \substitution{P}{u}{x} * \code{x} = \code{r.v}}}
%
    %
    \and
    \inferhref{Store}{Store}{}
    {\CSL{\acc{\code{r.v}}}{\code{r.v := e}}{ \acc{\code{r.v}} * \code{r.v} = \code{e} }}
    \and
    \inferhref{Free}{Free}{}
    {\CSL{\acc{\code{q.v}}}{\fefree{\code{q}}}{ \top }}
    \and
    \inferH{Assign}{}
    {\CSL{\substitution{P}{x}{e}}{\code{x := e}}{P}}
    \and
    \inferH{Skip}{}
    {\CSL{P}{\feskip}{P}}
\end{mathpar}
\caption{Inference rules of our IDF-based CSL.}
\label{fig:csl}
\end{figure}

\paragraph{An IDF-based program logic for \parImp{}}
We build and prove sound a program logic analogous to CSL for ParImp based on our IDF state model $\Sigma_{\mathit{IDF}}$ (defined in \secref{subsec:idf-algebra}). Our framework also supports standard separation logic, but connecting an IDF logic to \ViperInst{} allows us to focus on the most interesting aspects of the soundness proof.

Our program logic judgment is written $\CSL{P}{C}{Q}$, where $P$ and $Q$ are \ViperInst{} assertions (\ie sets of IDF states).
The most important rules of our program logic are given in \figref{fig:csl}.
The rules \ruleref{Seq2}, \ruleref{Cons2}, \ruleref{If}, \ruleref{While}, \ruleref{Free},
\ruleref{Assign}, and \ruleref{Skip}, are standard.
The rules \ruleref{Alloc}, \ruleref{Store}, \ruleref{Load}, \ruleref{Frame2}, and \ruleref{Par3} are analogous to the standard CSL rules, but adapted to our IDF setting.
In particular, the rule \ruleref{Frame2} requires the precondition $P$ and the frame $F$ to be self-framing.
Without this restriction, one could for example use $P \triangleq (\acc{\code{r.v}})$ and $F \triangleq (\code{r.v} = 5)$ to unsoundly derive the invalid triple
$\CSL{(\acc{\code{r.v}}) * (\code{r.v} = 5)}{\code{r.v := 3}}{(\acc{\code{r.v}} * \code{r.v} = 3) * (\code{r.v} = 5)}$
(whose postcondition is not satisfiable)
using the rules \ruleref{Frame2} and \ruleref{Store}.
Similarly, the rule \ruleref{Par3} requires the preconditions $P_l$ and $P_r$ to be self-framing.
Finally, the rule \ruleref{Load} allows arbitrary preconditions $P$, as long as $P$ asserts some permission to read \code{r.v}.

We have proved in Isabelle the soundness of this IDF-based program logic, which we state as follows
(the proof of this theorem is an adaption of the proof from \citet{VafeiadisCSLSoundness} to our IDF setting):
\begin{theorem}[Adequacy]\label{thm:adequacy}
    Let $C$ be a well-typed program, and $P$ and $Q$ be predicates on \parImp{} states (\ie without permissions).
    If the triple $\CSL{ P }{C}{Q}$ holds,
    and if $\sigma$ is a well-typed state such that $P(\sigma)$,
    then executing $C$ in the state $\sigma$ will not abort nor encounter any data race, and
    for all $\sigma'$ such that $\langle C, \sigma \rangle \rightarrow^* \langle \feskip, \sigma' \rangle$, we have $Q(\sigma')$.
\end{theorem}

\subsection{A Sound Front-End Translation}
\label{subsec:translation}

\begin{figure}
\footnotesize
\begin{equation*}
    \begin{aligned}
        &\translate{r := \fenew{e}} &&\triangleq ((\vhavoc{r}; \vinhale{\acc{r.v} * r.v = e}), \varnothing) \\
        &\translate{ \fefree{r} } &&\triangleq (\vexhale{\acc{r.v}}, \varnothing) \\
        &\translate{ C_1; C_2 } &&\triangleq ((\translate{C_1}.1; \translate{C_2}.1), (\translate{C_1}.2 \cup \translate{C_2}.2)) \\
        &\translate{ \feif{b}{C_1}{C_2} } &&\triangleq (\vif{b}{\translate{C_1}.1}{\translate{C_2}.1}), (\translate{C_1}.2 \cup \translate{C_2}.2))
    \end{aligned}
    \quad
    \quad
    \begin{aligned}
        &\translate{\feskip} &&\triangleq (\vskipp, \varnothing) \\
        &\translate{x := e} &&\triangleq (x := e, \varnothing) \\
        &\translate{r.v := e} &&\triangleq (r.v := e, \varnothing) \\
        &\translate{x := r.v} &&\triangleq (x := r.v, \varnothing)
    \end{aligned}
\end{equation*}
\begin{align*}
    &\translate{ C_l \parcomp C_r } &\triangleq \; &(
        (\vexhale{P_l * P_r}; \vhavoc{\modvars{C_l} \cup \modvars{C_r}}; \vinhale{Q_l * Q_r}), \\
        &&&\{ \vinhale{P_l}; \translate{C_l}.1; \vexhale{Q_l} \} \cup \{ \vinhale{P_r}; \translate{C_r}.1; \vexhale{Q_r} \} \cup
        \translate{C_l}.2 \cup \translate{C_r}.2
    ) \\
    &\translate{ \fewhile{b}{C} } &\triangleq \; &(
        (\vexhale{I}; \vhavoc{\modvars{C}}; \vinhale{I \land \lnot b}), \\
        &&&\{ \vinhale{I \land b}; \translate{C}.1; \vexhale{I} \} \cup \translate{C}.2
    )
\end{align*}
\caption{Front-end translation from \parImp{} to \ViperInst{}.
The translation function $\translate{\_}$ takes as input an annotated \parImp{} statement $C$
and returns a pair of a \ViperInst{} statement and a set of \ViperInst{} statements.
We write $\fst{\translate{C}}$ and $\snd{\translate{C}}$ to denote its first and second components, respectively.
Assertions $P_l$, $P_r$, $Q_l$, and $Q_r$ for the parallel composition
and $I$ for the while loop are annotations provided by the user, which are all required to be self-framing.
The notation $\vhavoc{V}$, where $V$ is a set of variables $\{ x_1, \ldots, x_n \}$, is a shorthand
for $\vhavoc{x_1}; \ldots; \vhavoc{x_n}$.}
\label{fig:translation}
\end{figure}

Building on the previously-defined IDF-based program logic,
we define a standard front-end translation from \parImp{} programs with annotations into \ViperInst{} programs, shown in \figref{fig:translation}.
This translation was illustrated in the example in \figref{fig:running-example} from \secref{sec:key-ideas}.
The translation function $\translate{\_}$ takes as input an annotated \parImp{} statement $C$
and yields a \emph{pair of a \ViperInst{} statement and a set of \ViperInst{} statements}.
The first component, written $\fst{\translate{C}}$, corresponds to the main translation of $C$,
while the second component, written $\snd{\translate{C}}$, corresponds to the set of
auxiliary Viper methods generated by the translation along the way.
Auxiliary methods are generated for loops and parallel compositions only.
Methods \code{l} and \code{r} in \figref{fig:running-example} are examples of such auxiliary methods.

The translation of field and variable assignments is straightforward.
The translation of sequential composition and conditional statements is also straightforward since they use the corresponding sequential composition and conditional statements of \ViperInst{}, and collect the auxiliary methods generated by the translation of the sub-statements.
The translation of allocation and deallocation statements corresponds to the rules \ruleref{Alloc} and \ruleref{Free} from \figref{fig:csl}.

The translation of parallel composition and while loops is more involved, but they follow the same pattern.
First, the premises of the relevant rules (\ruleref{Par3} and \ruleref{While}) are checked by generating auxiliary methods, which first inhale the relevant precondition, then translate the relevant statement, and finally exhale the relevant postcondition.
For example, the premise
$\CSL{I \land b}{C}{I}$
of the rule \ruleref{While}
is checked by generating the auxiliary method $\vinhale{I \land b}; \translate{C}.1; \vexhale{I}$.
We call this pattern the \emph{inhale-translation-exhale} pattern.
Then, the main translation follows the conclusion of the rule, by exhaling the precondition, havocking the modified variables, and inhaling the postcondition.
For example, the main translation of the loop $\fewhile{b}{C}$ is $(\vexhale{I}; \vhavoc{\modvars{C}}; \vinhale{I \land \lnot b})$, reflecting the conclusion
$\CSL{I}{\fewhile{b}{C}}{I \land \lnot b}$
of the rule \namerule{While}.
We call this pattern,
which we have already seen in \secref{subsec:ax-semantics-key},
the \emph{exhale-havoc-inhale} pattern.
Those two patterns are not specific to our translation, but are general patterns that can be found in many front-end translations.

\paragraph{Soundness}
We assume that the \parImp{} statement $C$ we want to verify is annotated with a precondition $P$ and a postcondition $Q$.
In this case, we add $\vinhale{P}$ before the main translation (as we did in \figref{fig:running-example}), and $\vexhale{Q}$ afterwards,
following the inhale-translation-exhale pattern.
Our complete front-end translation yields the set of ViperCore statements $\{ \vinhale{P}; \fst{\translate{C}}; \vexhale{Q} \} \cup \snd{\translate{C}}$.
Our translation is sound, as stated in the following theorem.
We say that a ViperCore statement $C_v$ is valid \emph{w.r.t.\ the axiomatic semantics},
which we write $\mathit{valid}_{Ax}(C_v)$, iff $\exists B \ldotp \axiomSem{\top}{C_V}{B}$.

\begin{theorem}[Soundness of the front-end translation]\label{thm:fe-soundness}
    Let $C$ be a front-end statement, and $P$ and $Q$ be assertions.
    If (1) the axiomatic semantics triple $\axiomSem{P}{\fst{\translate{C}}}{Q}$ holds,
    and (2) all \thibault{ViperCore} statements in $\translate{C}.2$ are valid \wrt{} the axiomatic semantics,
    then $\CSL{P }{C}{Q }$ holds.
\end{theorem}

\newcommand{\isValid}[1]{\mathit{valid}(#1)}
\newcommand{\convertible}[1]{\mathit{convertible(#1)}}

To prove this theorem, we show that the translation of every front-end statement $C$ is \emph{backward-convertible} (or \emph{convertible} in short),
which we write as $\convertible{C}$.
Intuitively, this means that if the translation of the front-end statement into \ViperInst{} is valid (including all auxiliary \ViperInst{} methods)
then we can convert the axiomatic semantics triple $\CSL{P}{\fst{\translate{C}}}{Q}$ into a front-end triple $\CSL{P}{C}{Q}$.
We formally express this property as follows:
$$\convertible{C} \triangleq \left(
    \forall P, Q \ldotp
    ((\forall C_v \in \snd{\translate{C}} \ldotp \mathit{valid}_{Ax}(C_v)) \land
    \axiomSem{P}{\fst{\translate{C}}}{Q}) \Rightarrow
    \CSL{P}{C}{Q}
\right)
$$
This convertibility property combined with the following lemma allows us to prove \thmref{thm:fe-soundness}:
\begin{lemma}[Inhale-translation-exhale pattern]
    \label{lem:inhale-translation-exhale}
    If (1) $\convertible{C}$ holds,
    (2) all auxiliary methods from $\snd{\translate{C}}$ are valid \wrt{} the axiomatic semantics,
    and (3) $\axiomSem{P}{ \vinhale{A}; \fst{\translate{C}}; \vexhale{B} }{ Q }$ holds,
    then $\CSL{P * A}{C}{B * Q}$ holds.
\end{lemma}

\begin{proof}
    By inverting the rules \axNameFull{Seq}, \axNameFull{Inhale}, and \axNameFull{Exhale},
    we get the existence of $R$ such that
    (a) $\axiomSem{P * A}{ \fst{\translate{C}} }{ R }$ holds
    and (b) $R \models B * Q$.
    By applying $\convertible{C}$, and from (2) and (a), we get $\CSL{P * A}{C}{R}$.
    We conclude by combining (b) with the rule \ruleref{Cons2}.
\end{proof}

The proof of this lemma is straightforward thanks to \theIVL{}'s \emph{axiomatic} semantics. Relating CSL to an operational IVL semantics would require substantially more effort to re-prove standard reasoning principles, which we prove once and for all in the equivalence proof of the two IVL semantics.

We now need to prove $\convertible{C}$ for all $C$, which we do by structural induction.
The inductive cases for most statements are straightforward;
the interesting cases are allocation, deallocation, parallel compositions, and while loops.
As explained above, the main translation of those statements follows the same exhale-havoc-inhale pattern, which we have already seen in \secref{subsec:ax-semantics-key}, and prove below:
\begin{lemma}[Exhale-havoc-inhale]
    \label{lem:exhale-havoc-inhale}
	Let $P$ and $Q$ be self-framing assertions.\footnote{This condition is trivially true for standard SLs.}
  Assume that\\
  $\axiomSem{A}{\vexhale{P}; \vhavoc{x_1}; \ldots; \vhavoc{x_n}; \vinhale{Q}}{B}$ holds,
  where $\{ x_1, \ldots, x_n \} = \modvars{C}$.
  If $\anySL{P}{C}{Q}$ holds, and if
  $\mathcal{L}$ has a frame rule and a consequence rule, then
  $\anySL{A}{C}{B}$ holds.
\end{lemma}

\begin{proof}
    By inverting the rule \axNameFull{Seq}, we obtain $F$ such that
  (a) $\axiomSem{F}{\vinhale{Q}}{B}$ and
  (b) $\axiomSem{A}{\vexhale{P}; \vhavoc{x_1}; \ldots; \vhavoc{x_n}}{F}$ hold.
  From (b), by inverting the rules \axNameFull{Seq} and \axNameFull{Havoc}, we obtain an assertion $R$
  such that
  (c) $\axiomSem{A}{\vexhale{P}}{R}$ holds,
  (d) $\freevars{F} \cap \{ x_1, \ldots, x_n \} = \emptyset$, and
  (e) $R \models F$.%
  \footnote{More precisely,
  we obtain $F = (\exists x_1, \ldots, x_n \ldotp R)$, from which (d) and (e) follow.}
  By applying the frame rule with $F$ and $\anySL{P}{C}{Q}$, where the side condition is justified by (d), we get $\anySL{P * F}{C}{Q * F}$.
  Finally, we obtain $B = F * Q$ from (a) (by inverting the rule \axNameFull{Inhale}),
  and $A \models P * F$ from (c) (by inverting the rule \axNameFull{Exhale}) and (e);
  applying the consequence rule yields $\anySL{A}{C}{B}$.
\end{proof}

This proof shows that, in this pattern, the role of the \vexhaleNa{} statement, followed by a sequence of \vhavocNa{} statements, is to compute (implicity) the suitable frame for the front-end statement.
The \vinhaleNa{} statement afterwards then adds the postcondition of the front-end statement to the frame.

$\convertible{\fefree{r}}$
and $\convertible{\code{r} := \fenew{\code{e}}}$ follow directly from the lemma above,
by observing that $\vinhale{\top}$ and $\vexhale{\top}$ are equivalent to $\vskipp$ (and so omitted when encoding).

To prove $\convertible{\fewhile{b}{C}}$ (assuming $C$ is convertible),
we first apply \lemref{lem:inhale-translation-exhale} on the auxiliary method $\vinhale{I \land b}; \translate{C}.1; \vexhale{I}$
to get $\CSL{I \land b}{C}{I}$.
We then apply the rule \ruleref{While} to get $\CSL{I}{\fewhile{b}{C}}{I \land \lnot b}$.
Finally, we conclude by applying \lemref{lem:exhale-havoc-inhale} on the main translation $(\vexhale{I}; \vhavoc{\modvars{C}}; \vinhale{I \land \lnot b})$.

The proof of $\convertible{C_1 || C_2}$ proceeds similarly,
by first applying \lemref{lem:inhale-translation-exhale} on the two auxiliary methods (corresponding to the two premises of the rule \ruleref{Par3}), then applying the rule \ruleref{Par3},
and concluding by applying \lemref{lem:exhale-havoc-inhale}.
This concludes the proof of $\convertible{C}$ for all $C$, and thus the proof of \thmref{thm:fe-soundness}.

\paragraph{Summary}
We have demonstrated how the axiomatic semantics from \secref{subsec:ax-semantics} helps us solve Challenge 3, by allowing us to prove general lemmas about patterns that are common in front-end translations in a simple and straightforward manner, and to prove the soundness of a concrete front-end translation for a parallel programming language with multiple features not present in the IVL (\eg loops, dynamic memory allocation and deallocation).
\section{Related Work}
\label{sec:related-work}

\paragraph{Semantics of SL-based IVLs.}
There are two recent formalizations~\cite{THSem,GradualViper} of  subsets of Viper~\cite{Viper}.
However, each of them exposes implementation details of a Viper back-end, which does not allow the semantics to be connected to diverse back-ends and also not easily to front-ends.
In particular, \citet{THSem} use a total heap representation reflecting the Viper VCG back-end that translates to Boogie (as discussed in~\secref{sec:verif-cond-gener}), and \citet{GradualViper} reflect Viper's symbolic execution back-end.

GIL~\cite{GillianII}, which is the intermediate language of Gillian~\cite{GillianI,GillianII}, is parametric in its (1)~state model, which must be provided as a PCM (supporting SL but not IDF states in contrast to \theIVL{}), (2) \emph{memory actions} operating on the state model, and (3) \emph{core predicates} describing atomic assertions on the memory such as a SL points-to assertion.
For each state instantiation, tool developers targeting GIL must specify \emph{produce} and \emph{consume} actions for each core predicate, which correspond to \vinhaleNa{} and \vexhaleNa{} operations in \theIVL{}.
Together with instantiated parameters, \citet{GillianII} provide an operational semantics for the symbolic execution of GIL\@.
Since the instantiated state effectively reflects the symbolic state on which the symbolic execution tool operates, a GIL instantiation essentially represents the back-end semantics.
This is in contrast to our \theIVL{}, which allows abstracting over multiple back-ends.

\citet{InliningViper} define the semantics of a parametric verification language similar to \theIVL{} for the purpose of showing formal results on method call inlining in automated SL verifiers.
Their semantics is meant to capture IVL \emph{back-ends} with their heuristics.
That is, an instantiation reflects a \emph{single} back-end.
As a result, in contrast to \theIVL{}, their semantics has no angelic nondeterminism.
Moreover, their notion of separation algebra to represent states does not support IDF.



\paragraph{Proofs connecting a front-end with an IVL}
\citet{SummersM20} and \citet{WolfSM22} reason about the correctness of translations into a SL-based IVL by providing proof sketches for mapping a correct Viper program to a proof for Hoare triples in the RSL weak memory logic~\cite{VafeiadisN13} and the TaDa logic~\cite{PintoDG14}, respectively.
However, the reasoning is done via proof sketches on paper, which explore only the high-level reasoning principles and thus avoid many of the complexities involved in a fully formal proof.
Neither of these works formally reasons about the underlying Viper semantics; they describe the behavior of Viper encodings informally.

\citet{GillianTR} briefly describe a parametric soundness framework for GIL (the intermediate language of Gillian~\cite{GillianI,GillianII}).
They show that if certain conditions hold on the instantiations of the GIL parameters, then the resulting symbolic execution is sound \wrt{} a concretization function on symbolic states.
However, they do not provide an IVL semantics like \theIVL{} that abstracts uniformly over multiple back-ends.
Additionally, since GIL does not support concurrency~\cite{GillianI,GillianII},
their soundness framework cannot reason about the encoding of front-end languages such as \parImp{} described in \autoref{sec:frontends}.
\citet{GillianSoundnessECOOP} present a formal compositional symbolic execution engine inspired by Gillian. In contrast to our work, they focus on supporting both over-approximating and under-approximating reasoning, and do not model an IVL, but only apply their framework to a simple front-end language with a fixed memory model.

There is also work proving the soundness of front-end translations to IVLs not based on SL~\cite{VogelsJP09,Backes2011,Herms13,Fortin13a,THSem}.
However, in contrast to our setting, the corresponding translations do not reflect rules in a front-end program logic.
As a result, the soundness proofs work naturally at the level of an operational semantics for the front-end and IVL\@.
Examples include translations from the Dminor data processing language to the Bemol IVL~\cite{Backes2011}, from C to the WhyCert IVL (inspired by the Why3 IVL)~\cite{Herms13}, and from Viper to Boogie~\cite{THSem} (in the case of the Viper-to-Boogie translation, Viper is the front-end and Boogie is the target IVL).

\paragraph{Proofs connecting an IVL with a back-end.}
%
\citet{THSem} show the soundness of the Viper back-end that translates to Boogie.
In our work, we show that their back-end specific semantics respects our more generic version (\secref{sec:verif-cond-gener}).
The work most closely related to the symbolic execution back-end presented in \autoref{sec:symbolic-execution} is \citet{GradualViper}'s formalization of a variant of Viper's symbolic execution back-end targeted at gradual verification.
Due to their focus on gradual verification, they only target a simplified model of Viper that (unlike our symbolic execution) does not support fractional permission.
As a consequence, they can use a simpler implementation that does not rely on continuation passing style and they can ignore some of the complexities described in \autoref{sec:symbolic-execution} such as state consolidation.
Also they formalize the symbolic execution via a derivation tree, while we implement it as an \isabelle{} function.
\citet{FeatherweightVerifast} prove a formalization of VeriFast's symbolic execution sound.
Compared to our work, they do not have a semantics that captures different verification algorithms, or supports IDF or fractional permissions.

There is also work on non SL-based IVL back-end proofs. These back-ends typically have simple state models and use different algorithms compared to SL-based back-ends.
For example, \citet{ParthasarathyMuellerSummers21} generate soundness proofs for Boogie's VCG, and \citet{VogelsJP10} prove a VCG for a similar IVL sound once and for all.
\citet{Garchery21} and \citet{Cohen24} validate certain logical transformations performed in the Why3 IVL verifier.

\paragraph{Angelic non-determinism}
Angelic non-determinism \cite{FloydAngNondet} has been widely used from encoding partial programs~\cite{AngelicPartialPrograms}, to representing interaction between code written in multiple languages~\cite{DimSum, Melocoton}, to encoding specifications~\cite{FloydAngNondet, CCR}.
However, to the best of our knowledge, our work is the first to use angelic non-determinism to abstract over different verification algorithms.
\citet{FeatherweightVerifast} and \citet{CCR} both also use angelism for \vexhaleNa{}, but do not abstract over or formally connect with diverse back-end algorithms, as we do.
Instead, \citet{FeatherweightVerifast} use angelism to represent a symbolic execution algorithm, while \citet{CCR} use angelism to encode the transfer of resources in a refinement calculus.

\paragraph{Implicit dynamic frames (IDF)}
IDF was originally presented with a fixed resource model (\ie{} full ownership to a heap location) and where the heap is represented as a \emph{total} mapping from heap locations to values~\cite{SmansJP12}.
\citet{ParkinsonSummers12} formally showed the relationship between IDF and SL by defining a logic based on \emph{total} heaps and separate permission masks that captures both.
They also consider fixed resource models of IDF and SL (\ie{} fractional ownership to a heap location~\cite{Boyland03}).
Our work generalizes the notion of a separation algebra~\cite{Calcagno2007,Dockins2009} to capture arbitrary resource models for IDF \emph{and} SL in the same framework.
In particular, the algebra does not fix a particular state representation.
This enables, for instance, a \emph{partial} heap instantiation for IDF that we use to formalize Viper's state model (\secref{sec:viper-inst}).
SteelCore~\cite{SwamyRFMAM20} is a framework
with an extensible CSL to reason about concurrent F*~\cite{FStar} programs.
The extensibility of the framework is in particular demonstrated by allowing IDF-style preconditions of the restricted form $P * b$ (compared to the more general IDF assertions supported in our work), where $P$ is an SL assertion, and $b$ is a heap-dependent boolean expression framed by $P$ (and similarly for postconditions).

\paragraph{Other approaches}
In this paper, we showed how one can formally establish the soundness of translational SL verifiers, but there are also other approaches to building automated SL verifiers and establishing their soundness.
Steel~\cite{Steel} is an SL-based proof-oriented programming language in F*.
Steel programs are automatically proved correct using a type checker that is proved sound against SteelCore; the type checker uses an SMT solver to discharge proof obligations.

\michael{\citet{Katamaran} (building on the ideas of \citet{FeatherweightVerifast}) show how to build a verified symbolic execution based on a specification monad that allows expressing angelic and demonic non-determinism and assume and assert statements. They \peter{formalize (in Coq)} two (structurally identical) versions of the symbolic execution algorithm: a deeply embedded version that allows execution and a shallow embedded version \peter{to prove soundness of the former.} However, both versions represent the same algorithm; they do not consider different back-ends (like the verification condition generation back-end in \autoref{sec:verif-cond-gener}).}

\citet{RefinedC} propose an approach to building sound verifiers that requires writing the verifier in a domain specific language called Lithium.
Verifiers in Lithium can be automatically executed inside the Coq proof assistant and produce a foundational proof of correctness.
Lithium-based verifiers are not translational, but work directly on the source-language program.
\section{Conclusion}
\label{sec:conclusion}

We have presented a formal framework for reasoning about the soundness of translational separation logic verifiers.
We have defined an abstract IVL, whose state model can be instantiated with any \IDFAlgebra{}.
An operational and an equivalent axiomatic semantics allow one to connect the IVL to back-ends and front-ends, resp.
Crucially, the semantics leverage dual non-determinism to capture different verification algorithms implemented by different back-end verifiers.
We have illustrated the usefulness of our formal framework by instantiating it with elements of Viper,
connecting it to two Viper back-ends, and using it to prove soundness of a front-end translation for an IDF-based concurrent separation logic.
The main direction for future work is to use our formal framework to model additional IVLs and prove soundness of complex translational verifiers.

\begin{acks}
  We thank Ellen Arlt and Hongyi Ling for their useful feedback on the framework presented in this paper.
  This work was partially funded by the
  \grantsponsor{SNSF}{Swiss National Science Foundation (SNSF)}{http://dx.doi.org/10.13039/501100001711}
  under Grant No.~\grantnum{SNSF}{197065}.
\end{acks}

\section*{Data availability statement}
All technical results presented in this paper have been formalized and proven in Isabelle/HOL, and our formalization is publicly available in our artifact~\cite{artifact}. The development version is available at \url{https://github.com/viperproject/viper-roots}.

{
\interlinepenalty=10000
\bibliography{bib}
}

\appendix

\ifextended
\clearpage
\section{Full Definition of Symbolic Execution}
\label{sec:sexec-full}

The main functions of our symbolic execution are the $\sexec$, $\sproduce$, $\sconsume$, and $\sexp$ functions, whose definition is given below.

$\sexec~\SState~\stmt~\SCond$ symbolically executes the statement $\stmt$ in the symbolic state $\SState$:
{\small
    \begin{align*}
    \sexec~\SState~\stmt~\SCond &\eqdef \left\{
    \begin{aligned}
      &\sproduce~\SState~\assertion~\SCond & \text{ if } \stmt = \Inhale{\assertion}\\
      &\sconsume~\SState~\assertion~(\Lam \SState. \sstablize~\SState~\SCond) & \text{ if } \stmt = \Exhale{\assertion}\\
      &\sexp~\SState~\expr~(\Lam \SState~\symexpr. \sexec~\symaddcond{\SState}{\symexpr}~\stmt_1~\SCond & \text{ if } \stmt = (\SynIf{\expr}{\stmt_1}{\stmt_2})
      \\ &\quad \wedge \sexec~\symaddcond{\SState}{\neg\symexpr}~\stmt_2~\SCond)
     \\
      &\sexec~\SState~\stmt_1~(\Lam \SState.\sexec~\SState~\stmt_2~\SCond) & \text{ if } \stmt = {\stmt_1} \cseq {\stmt_2}\\
      &x \in \SState.\symstore \wedge \sexp~\SState~\expr~(\Lam \SState~\symexpr.\SCond~\SState\{\symstore :=\SState.\symstore[x\mapsto\symexpr]\}) & \text{ if } \stmt = \vassign{x}{\expr}\\
      &x \in \SState.\symstore \wedge \SCond~\SState\{\symstore :=\SState.\symstore[x\mapsto\symfresh]\}) & \text{ if } \stmt = \vhavoc{x}\\
      &\sexp~\SState~\expr_r~(\Lam \SState~\symexpr_r.
       \sexp~\SState~\expr_v~(\Lam \SState~\symexpr_v.
        \symextract{\SState}{\symexpr_r}{\fieldname}{1}~(\Lam \SState~\chunk.
              & \text{ if } \stmt = \vfieldassign{\expr_r}{\fieldname}{\expr_v} \\
        &\quad\sstablize~\SState~(\Lam \SState.
        \symaddchunk{\SState}{\chunk\{\chunkval := \symexpr_v\}}~\SCond))))
    \end{aligned}\right.
  \end{align*}
}

$\sproduce~\SState~\assertion~\SCond$ inhales the assertion $\assertion$ in the symbolic state $\SState$.
{\small
    \begin{align*}
    \sproduce~\SState~\assertion~\SCond &\eqdef \left\{
    \begin{aligned}
      &\sexp~\SState~\expr~(\Lam \SState~\symexpr. \SCond~\symaddcond{\SState}{\symexpr}) & \text{ if } \assertion = \expr\\
      &\sexp~\SState~\expr_r~(\Lam\SState~\symexpr_r.
        \sexp~\SState~\expr_p~(\Lam\SState~\symexpr_p. & \text{ if } \assertion = \accmath{\expr_r}{\fieldname}{\expr_p}\\
      & \quad \symaddchunk{\SState}{\{\symexpr_r, \fieldname, \symexpr_p, \symfresh\}}~\SCond))\\
      &\sexp~\SState~\expr_r~(\Lam\SState~\symexpr_r.
        \mathsf{let}~\symexpr_p \mathop{:=} \symfresh~\mathsf{in}
        & \text{ if } \assertion = \accmath{\expr_r}{\fieldname}{\wildcardmath}\\
      & \quad \symaddchunk{\symaddcond{\SState}{0 < \symexpr_p}}{\{\symexpr_r, \fieldname, \symexpr_p, \symfresh\}}~\SCond))\\
      &\sproduce~\SState~\assertion_1~(\Lam\SState.
        \sproduce~\SState~\assertion_2~\SCond) & \text{ if } \assertion = \assertion_1 \ast \assertion_2\\
      &\sexp~\SState~\expr~(\Lam \SState~\symexpr. \sproduce~\symaddcond{\SState}{\symexpr}~\assertion'~\SCond & \text{ if } \assertion = \expr \Rightarrow \assertion'
      \\ &\quad \wedge \SCond~\symaddcond{\SState}{\neg\symexpr}) \\
      &\sexp~\SState~\expr~(\Lam \SState~\symexpr. \sproduce~\symaddcond{\SState}{\symexpr}~\assertion_1~\SCond & \text{ if } \assertion = (\expr \mathop{?} \assertion_1 \mathop{:} \assertion_2)
      \\ &\quad \wedge \sproduce~\symaddcond{\SState}{\neg\symexpr}~\assertion_2~\SCond)
    \end{aligned}\right.
  \end{align*}
}

$\sconsume~\SState~\assertion~\SCond$ exhales the assertion $\assertion$ in the symbolic state $\SState$.
{\small
    \begin{align*}
    \sconsume~\SState~\assertion~\SCond &\eqdef \left\{
    \begin{aligned}
      &\sexp~\SState~\expr~(\Lam\SState~\symexpr. (\symproves{\SState}{\symexpr} )\wedge \SCond~\SState) & \text{ if } \assertion = \expr\\
      &\sexp~\SState~\expr_r~(\Lam\SState~\symexpr_r.
        \sexp~\SState~\expr_p~(\Lam\SState~\symexpr_p. \symextract{\SState}{\symexpr_r}{\fieldname}{\symexpr_p}~ & \text{ if } \assertion = \accmath{\expr_r}{\fieldname}{\expr_p}\\
      & \quad        (\Lam\SState~\chunk. \symaddchunk{\SState}{\chunk\{\chunkperm :=\chunk.\chunkperm - \symexpr_p\}}~\SCond)))\\
      &\sexp~\SState~\expr_r~(\Lam\SState~\symexpr_r.
        \symextract{\SState}{\symexpr_r}{\fieldname}{\wildcardmath}~ & \text{ if } \assertion = \accmath{\expr_r}{\fieldname}{\wildcardmath}\\
      & \quad        (\Lam\SState~\chunk. \symaddchunk{\SState}{\chunk\{\chunkperm :=\chunk.\chunkperm / 2\}}~\SCond))\\
      &\sconsume~\SState~\assertion_1~(\Lam \SState.
        \sconsume~\SState~\assertion_2~\SCond) & \text{ if } \assertion = \assertion_1 \ast \assertion_2\\
      &\sexp~\SState~\expr~(\Lam \SState~\symexpr. \sconsume~\symaddcond{\SState}{\symexpr}~\assertion'~\SCond & \text{ if } \assertion = \expr \Rightarrow \assertion'
      \\ &\quad \wedge \SCond~\symaddcond{\SState}{\neg\symexpr}) \\
      &\sexp~\SState~\expr~(\Lam \SState~\symexpr. \sconsume~\symaddcond{\SState}{\symexpr}~\assertion_1~\SCond & \text{ if } \assertion = (\expr \mathop{?} \assertion_1 \mathop{:} \assertion_2)
      \\ &\quad \wedge \sconsume~\symaddcond{\SState}{\neg\symexpr}~\assertion_2~\SCond)
    \end{aligned}\right.
  \end{align*}
  }

$\sexp~\SState~\expr~\SCond$ symbolically evaluates the expression $\expr$ in the symbolic state $\SState$. (This definition has been slightly simplified by removing the treatment of lazy binary operators like \verb|&&| or \verb/||/.)
{\small
    \begin{align*}
    \sexp~\SState~\expr~\SCond &\eqdef \left\{
    \begin{aligned}
      &\SCond~\SState~\literal & \text{ if } \expr = \literal\\
      &x\in\SState.\symstore \wedge \SCond~\SState~\SState.\symstore[x] & \text{ if } \expr = x\\
      &\sexp~\SState~\expr'~(\Lam \SState~\symexpr. \SCond~\SState~(\unoperation\symexpr)) & \text{ if } \expr = \unoperation\expr'\\
      &\sexp~\SState~\expr_1~(\Lam \SState~\symexpr_1.
        \sexp~\SState~\expr_2~(\Lam \SState~\symexpr_2. \SCond~\SState~(\symexpr_1\binoperation\symexpr_2))) & \text{ if } \expr = \expr_1\binoperation\expr_2\\
      &\sexp~\SState~\expr'~(\Lam \SState~\symexpr. \sexp~\symaddcond{\SState}{\symexpr}~\expr_1~\SCond & \text{ if } \expr = (\expr' \mathop{?} \expr_1 \mathop{:} \expr_2)
      \\ &\quad \wedge \sexp~\symaddcond{\SState}{\neg\symexpr}~\expr_2~\SCond) \\
      &\sexp~\SState~\expr_r~(\Lam\SState~\symexpr_r.
\symextract{\SState}{\symexpr_r}{\fieldname}{0}~(\Lam\SState~\chunk. & \text{ if } \expr = \expr_r.\fieldname\\
      & \quad \symaddchunk{\SState}{\chunk}~(\Lam\SState.\SCond~\SState~\chunk.\chunkval)))
    \end{aligned}\right.
  \end{align*}
}

\fi

\end{document}